\documentclass[a4paper,USenglish]{lipics-v2021}
\usepackage{amsmath,amssymb,color,enumerate,graphicx,hyperref,microtype,tabularx,wrapfig,xspace}
\nolinenumbers

\renewcommand{\leq}{\leqslant}
\renewcommand{\geq}{\geqslant}
\newcommand{\mypara}[1]{\medskip\noindent{\sf\textbf{#1}}}
\newcommand{\mysmallpara}[1]{\medskip\noindent{\emph{#1}}}

\newcommand{\Reals}{\mathbb{R}}

\newcommand{\bd}{\partial}

\newcommand{\weight}{\mathrm{weight}}
\newcommand{\cost}{\mathrm{cost}}

\newcommand{\tree}{\ensuremath{\mathcal{T}}}

\newcommand{\graph}{\ensuremath{\mathcal{G}}}

\newcommand{\height}{\mathrm{height}}

\newcommand{\A}{\ensuremath{\mathcal{A}}}

\newcommand{\D}{\ensuremath{\mathcal{D}}}

\newcommand{\F}{\ensuremath{\mathcal{F}}}
\newcommand{\G}{\ensuremath{\mathcal{G}}}
\newcommand{\myH}{\ensuremath{\mathcal{H}}}
\newcommand{\myS}{\ensuremath{\mathcal{S}}}
\newcommand{\T}{\ensuremath{\mathcal{T}}}
\newcommand{\M}{\ensuremath{\mathcal{M}}}

\newcommand{\eps}{\varepsilon}

\newcommand{\etal}{\emph{et al.}\xspace}

\newcommand{\ig}{\graph^{\times}}

\newcommand{\gvis}{\ig_{\mathrm{vis},P}}

\newcommand{\sep}{\mathcal{\myS}}

\newcommand{\subdiv}{\M}

\newcommand{\poly}{\mathrm{poly}}

\newcommand{\mis}{\mbox{{\sc Maximum Independent Set}}\xspace}
\newcommand{\vc}{\mbox{{\sc Vertex Cover}}\xspace}
\newcommand{\fvs}{\mbox{{\sc Feedback Vertex Set}}\xspace}
\newcommand{\col}{\mbox{{\sc Coloring}}\xspace}
\newcommand{\nph}{\mbox{{\sc np}}-hard\xspace}
\newcommand{\mysup}{\mathrm{sup}}
\newcommand{\gsup}{\graph_{\mysup}}

\newcommand{\Vr}{V_{\mathrm{ref}}} 

\title{Clique-Based Separators for Geometric Intersection Graphs}

\funding{The work in this paper is supported by the Dutch Research Council (NWO) through Gravitation-grant NETWORKS-024.002.003.}

\author{Mark de Berg}{Department of Computer Science, TU Eindhoven, the Netherlands}{m.t.d.berg@tue.nl}{}{}
\author{S\'andor Kisfaludi-Bak\footnote{The research was conducted while the author was at the Max Planck Institute for Informatics, Saarbrücken, Germany.}}{Institute for Theoretical Studies, ETH Zürich, Switzerland}{sandor.kisfaludibak@eth-its.ethz.ch}{}{}
\author{Morteza Monemizadeh}{Department of Computer Science, TU Eindhoven, the Netherlands}{m.monemizadeh@tue.nl}{}{}
\author{Leonidas Theocharous}{Department of Computer Science, TU Eindhoven, the Netherlands}{l.theocharous@tue.nl}{}{}

\authorrunning{M.~de Berg and S.~Kisfaludi-Bak and M.~Monemizadeh and L.~Theocharous} 

\Copyright{Mark de Berg and S\'andor Kisfaludi-Bak and Morteza Monemizadeh and Leonidas Theocharous}

\EventEditors{John Q. Open and Joan R. Access}
\EventNoEds{2}
\EventLongTitle{42nd Conference on Very Important Topics (CVIT 2016)}
\EventShortTitle{CVIT 2016}
\EventAcronym{CVIT}
\EventYear{2016}
\EventDate{December 24--27, 2016}
\EventLocation{Little Whinging, United Kingdom}
\EventLogo{}
\SeriesVolume{42}
\ArticleNo{23}

\hideLIPIcs
\begin{document}
\maketitle

\begin{abstract}
Let $F$ be a set of $n$ objects in the plane and let $\ig(F)$ be its intersection graph.
A balanced clique-based separator of $\ig(F)$ is a set $\sep$ consisting of cliques
whose removal partitions $\ig(F)$ into components of size at most $\delta n$, for some 
fixed constant $\delta<1$. The weight of a clique-based separator is defined as
$\sum_{C\in\sep}\log (|C|+1)$.  
Recently De~Berg~\etal (SICOMP 2020) 
proved that if $S$ consists of convex fat objects, then $\ig(F)$ admits a balanced clique-based
separator of weight~$O(\sqrt{n})$. 
We extend this result in several directions, obtaining the following results.
\begin{itemize}
\item Map graphs admit a balanced clique-based separator of weight~$O(\sqrt{n})$, which is tight in the worst case. 
\item Intersection graphs of pseudo-disks admit a balanced clique-based separator of 
      weight~$O(n^{2/3}\log n)$. If the pseudo-disks are polygonal and of total complexity~$O(n)$
      then the weight of the separator improves to~$O(\sqrt{n}\log n)$. 
\item Intersection graphs of geodesic disks inside a simple polygon admit a balanced 
      clique-based separator of weight~$O(n^{2/3}\log n)$.
\item Visibility-restricted unit-disk graphs in a polygonal domain with $r$ reflex
      vertices  admit a balanced clique-based separator of weight~$O(\sqrt{n}+r\log(n/r))$, 
      which is tight in the worst case.
\end{itemize}
These results immediately imply sub-exponential algorithms for \mis (and, hence, \vc),
for \fvs, and for $q$-\col for constant~$q$ in these graph classes.
\end{abstract}

\setcounter{page}{1}

\section{Introduction}
The famous Planar Separator Theorem states that any planar graph $\graph=(V,E)$ with 
$n$~nodes\footnote{We use the terms \emph{nod}e and \emph{arc} when talking
about graphs, and \emph{vertex} and \emph{edge} for geometric objects.}
admits a subset $\sep\subset V$ of size~$O(\sqrt{n})$ nodes
whose removal decomposes $\graph$ into connected components of size at most~$2n/3$.
The subset~$\sep$ is called a balanced\footnote{For a separator to be balanced
it suffices that the components have size at most $\delta n$ for some constant~$\delta<1$.
When we speak of separators, we always mean balanced separators, unless
stated otherwise.} \emph{separator} of~$\graph$.
The theorem was first proved in 1979 by Lipton and Tarjan~\cite{LT-planar-separator-thm}, 
and it has been instrumental in the design of algorithms for planar graphs: it has
been used to design efficient divide-and-conquer algorithms, to design sub-exponential
algorithms for various \nph graph problems, and to design approximation algorithms for such
problems. 

The Planar Separator Theorem has been extended to various other graph classes.
Our interest lies in \emph{geometric intersection graphs}, where 
the nodes correspond to geometric objects 
and there is an arc between 
two nodes iff the corresponding objects intersect. If the objects are disks, the
resulting graph is called a \emph{disk graph}. Disk graphs, and in particular unit-disk graphs, 
are a popular model for wireless communication networks and have been studied extensively.
Miller~\etal~\cite{MTTV-sep-sphere-packing} 
and Smith and Wormald~\cite{SW-geom-sep} showed that if~$F$ is a set of balls in~$\Reals^d$ 
of ply at most~$k$---the \emph{ply} of $F$ is the maximum number of objects in~$F$ with a 
common intersection---then the intersection graph of $F$ has a separator of size~$O(k^{1/d}n^{1-1/d})$. 
This was generalized by Chan~\cite{C-PTAS-fat} and Har-Peled and Quanrud~\cite{HQ-low-dens-sep}
to intersection graphs of so-called low-density sets. 
Separators for string graphs---a string graph is an intersection graph of sets of curves in the plane---have 
also been
considered~\cite{FoxPT10,Lee-string-sep,Matousek14}, with
Lee~\cite{Lee-string-sep} showing that a separator of size~$O(\sqrt{m})$ exists, 
where $m$ is the number of arcs of the graph.
\medskip

Even for simple objects such as disks or squares, one must restrict the ply 
to obtain a separator of small size. Otherwise the objects can form a single clique, 
which obviously does not have a separator of sublinear size. To design subexponential algorithms 
for problems such as \mis, however, one can also work with a separator consisting
of a small number of cliques instead of a small number of nodes. Such \emph{clique-based separators}
were introduced recently by De~Berg~\etal~\cite{bbkmz-ethf-20}. Formally, a clique-based separator of a graph~$\graph$
is a collection $\sep$ of node-disjoint cliques 
whose union is a balanced separator
of~$\graph$. The \emph{weight} of $\sep$ is defined as 
$\weight(\sep) := \sum_{C\in \sep} \log(|C|+1)$.
De~Berg~\etal\cite{bbkmz-ethf-20} proved that the intersection graph of any set $F$ of
$n$ convex fat objects in the plane admits a clique-based separator of weight~$O(\sqrt{n})$, 
and they used this to obtain algorithms with running time ~$2^{O(\sqrt{n})}$ for many classic \nph problems on such graphs.
This running time is optimal, assuming the Exponential-Time Hypothesis (ETH).
The result generalizes to convex fat objects in $\Reals^d$, where the bound on the weight
of the clique-based separator becomes~$O(n^{1-1/d})$. 

The goal of our paper is to investigate whether similar results are possible
for non-fat objects in the plane.
Note that not all intersection graphs admit clique-based separators
of small weight. String graphs, for instance, can have arbitrarily large complete 
bipartite graphs as induced subgraphs, in which case any balanced clique-based separator has weight~$\Omega(n)$.

The first type of intersection graphs we consider are map graphs, which are a natural
generalization of planar graphs. The other types are generalizations of disk graphs. 
One way to generalize disk graphs is to consider fat objects instead of disks,
as done by De~Berg~\etal~\cite{bbkmz-ethf-20}. We will study three other generalizations,
involving non-fat objects: pseudo-disks, geodesic disks, and visibility-restricted unit disks.
Next we define the graph classes we consider more precisely; see
Fig.~\ref{fi:int-graphs}  for an example of each graph class. 
\begin{figure}
\begin{center}
\includegraphics[scale=1.2]{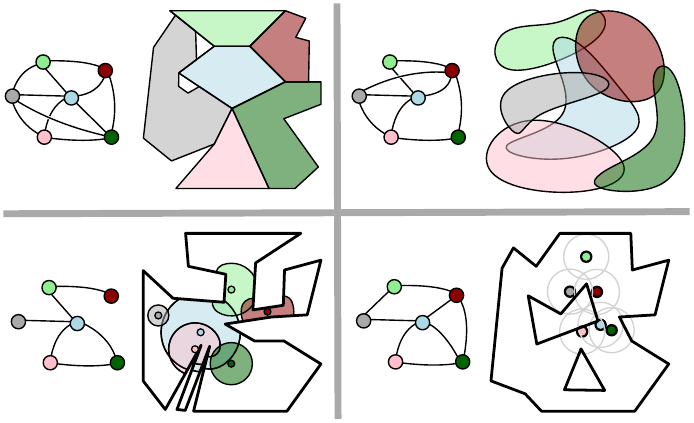}
\end{center}
\caption{A map graph, a pseudo-disk graph, a geodesic-disk graph,
and a visibility restricted unit-disk graph. For the latter class, the grey disks in the picture have radius~$\frac{1}{2}$.}
\label{fi:int-graphs}
\end{figure}

In the following, we use $\ig(F)$ to denote the intersection graph induced by a set $F$ of objects.
For convenience, we do not distinguish between the objects and the corresponding nodes,
so we use $F$ to denote the set of objects as well as the set of nodes in~$\ig(F)$.
We assume that the objects in $F$ are connected, bounded, and closed.

\mysmallpara{Map graphs.} 
Let $\subdiv$ be a planar subdivision and~$F$ be its set of faces. 
The graph with node set~$F$ that has an arc between every pair of neighboring faces is 
called the \emph{dual graph} of~$\subdiv$, and it is planar. 
Here two faces are neighbors if their boundaries have an edge of 
the subdivision in common.
A \emph{map graph}~\cite{CGP-map-graphs} is defined similarly, except now 
two faces are neighbors even if their boundaries meet in a
single point. Alternatively, we can define a map graph
as the intersection graph of a set $F$ of interior-disjoint regions in the plane.
Since arbitrarily many faces can share a vertex on their boundary, 
map graphs can contain arbitrarily large cliques. If 
at most $k$ faces meet at each subdivision vertex, the graph is called a
\emph{$k$-map graph}. Chen~\cite{C-map-graph-mis} proved that any $k$-map graph
has a (normal, not clique-based) separator of size~$O(\sqrt{kn})$, which is
also implied by Lee's recent result on string graphs~\cite{Lee-string-sep}.

\mysmallpara{Pseudo-disk graphs.} 
A set $F$ of objects is a set of pseudo-disks if for any~$f,f'\in F$ 
the boundaries $\bd f$ and $\bd f'$ intersect at most twice. 
Pseudo-disks were introduced in the context of motion planning by Kedem~\etal~\cite{KLPS-pseudo-disk-union},
who  proved that the union complexity of $n$ pseudo-disks is~$O(n)$.
Since then they have been studied extensively.
We  consider two types of pseudo-disks: polygonal pseudo-disks with $O(n)$
vertices in total, and arbitrary pseudo-disks.

\mysmallpara{Geodesic-disk graphs and visibility-restricted unit-disk graphs.}
As mentioned, unit-disk graphs are popular models for wireless communication networks.
We consider two natural generalizations of unit-disk graphs, which can be thought of
as communication networks in a polygonal environment that may obstruct communication. 
\begin{itemize}
\item \emph{Geodesic-disk graphs} in a simple polygon~$P$ 
      are intersection graphs of geodesic disks inside~$P$. (The \emph{geodesic disk} 
      with center $q\in P$ and radius~$r$ is the set of all points in~$P$ at geodesic distance at most~$r$
      from~$q$, where the geodesic distance between two points is the length of the shortest 
      path between them inside $P$.)
\item In \emph{visibility-restricted unit-disk graphs} the
      nodes correspond to a set $Q$ of $n$ points inside a polygon~$P$,
      which may have holes, and 
      two points $p,q\in Q$ are connected by an arc iff $|pq|\leq 1$
      and $p$ and $q$ see each other (meaning that $pq\subset P$).\footnote{Visibility-restricted
      unit-disk graphs are, strictly speaking, not intersection graphs. In particular, if 
      $R_q$ is defined as the region of points within $P$ that are visible from~$q$ and lie 
      within distance~$1/2$, then the visibility-restricted unit-disk graph is \emph{not} 
      the same as the intersection graph of the objects~$R_q$.} 
      A more general, directed version of such graphs was studied by
      Ben-Moshe~\etal~\cite{BHKM-vis-graph} under the name range-restricted visibility graph.
      They presented an output-sensitive algorithm to compute the graph.
\end{itemize}

\mypara{Our results: clique-based separator theorems.}
So far, clique-based separators were studied for fat objects:
De~Berg~\etal~\cite{bbkmz-ethf-20} consider convex or similarly-sized fat objects,
Kisfaludi-Bak~\etal\cite{KisfaludiMZ19} study how the fatness of
axis-aligned fat boxes impacts the separator weight,
and Kisfaludi-Bak~\cite{KB-hyperbolic-int-graphs} studies balls in hyperbolic space.
The $O(\sqrt{n})$ bound on the separator weight is tight even for unit-disk graphs. Indeed, a~$\sqrt{n}\times\sqrt{n}$ 
grid graph can be realized as a unit-disk graph, and any separator of such a grid graph
must contain $\Omega(\sqrt{n})$ nodes. Since the maximum clique size in a grid graph is two, 
any separator must contain $\Omega(\sqrt{n})$ cliques. All graph classes 
we consider can realize a~$\sqrt{n}\times\sqrt{n}$ grid graph, so $\Omega(\sqrt{n})$
is a lower bound on the weight of the clique-based separators we consider.
We obtain the following results.

In Section~\ref{sec:map-graphs} we show that any map graph has a clique-based separator of weight~$O(\sqrt{n})$. 
      This gives the first ETH-tight algorithms for \mis (and, hence, \vc), \fvs, and \col in map graphs; see below.

In Section~\ref{sec:pseudo-disks} we show that any intersection graph of pseudo-disks 
      has a clique-based separator of weight~$O(n^{2/3}\log n)$. 
      If the pseudo-disks are polygonal and of total complexity~$O(n)$
      then the weight of the separator improves to~$O(\sqrt{n}\log n)$. 
      
In Section~\ref{sec:geodesic-disks} we consider intersection graphs of geodesic 
      disks inside a simple polygon. At first sight, geodesic disks seem not much 
      harder to deal with than fat objects: 
      they can have skinny parts only in narrow corridors and then packing arguments may still be feasible. 
      Unfortunately another obstacle prevents us from applying a packing
      argument: geodesic distances in a simply connected polygon induce a metric space whose doubling
      dimension depends on the number of reflex vertices of the polygon. Nevertheless, by 
      showing that geodesic disks inside a simple polygon behave as pseudo-disks,
      we are able to obtain a clique-based separator of weight~$O(n^{2/3}\log n)$,
      independent of the number of reflex vertices.


In Section~\ref{sec:visibility} we study visibility-restricted unit-disk graphs.
      We give an $\Omega(\min(n,r\log(n/r)) +\sqrt{n})$ lower bound for the separator weight,
      showing that a clique-based separator whose weight depends only on~$n$, the number of points 
      defining the visibility graph, is not possible.
      We then show how to construct a clique-based separator
      of weight~$O(\min(n,r\log(n/r)) +\sqrt{n})$.

All separators can be computed in polynomial time. For map graphs and for the pseudo-disk intersection
graphs, we assume the objects have total complexity~$O(n)$. If the objects have curved
edges, we assume that basic operations (such as computing the intersection points of two such curves)
take $O(1)$~time.

\mypara{Applications.}
In Section~\ref{sec:applications} we apply our separator theorems to obtain subexponential algorithms for \mis, \fvs,
and $q$-\col for constant $q$ in the graph classes discussed above. The crucial property of these problems 
that makes our separator applicable, is that the possible ways in which a solution can 
``interact'' with a clique of size $k$ is polynomial in~$k$. We use known techniques (mostly from De Berg~\etal~~\cite{bbkmz-ethf-20}) 
to solve the three problems on any graph class that has small clique-based separators. 

All our graph classes are subsumed by string graphs. Bonnet and Rzazewski~\cite{BonnetR19}
showed that string graphs have $2^{O(n^{2/3}\log n)}$ 
algorithms for \mis and $3$-\col, and a $2^{n^{2/3}\log^{O(1)} n}$ algorithm 
for {\sc Feedback Vertex Set}, and that string graphs do not have subexponential algorithms for $q$-\col with $q\geq 4$ under ETH. One can also obtain subexponential algorithms in some of our classes from results of Fomin~\etal~\cite{FominLP0Z19,C-map-graph-mis} or Marx and Philipczuk~\cite{MarxP15}. 
The running times we obtain match or slightly
improve the results that can be obtained from these existing results.
It should be kept in mind, however, that the existing results are for more general graph classes.
An exception are our results on map graphs, which were explicitly studied before and where we
improve the running time for \mis and \fvs from $2^{O(\sqrt{n}\log n)}$ to $2^{O(\sqrt{n})}$.
(But, admittedly, the existing results apply in the parameterized setting while ours don't.)
In any case, the main advantage of our approach is that it allows us to solve \mis, \fvs and
$q$-\col on each of the mentioned graph classes in a uniform manner.
\section{Map graphs}
\label{sec:map-graphs}
Recall that a map graph is the intersection graph of a set~$F$ of interior-disjoint
objects in the plane. We construct a clique-based separator for $\ig(F)$ in
four steps. First, we construct a bipartite plane \emph{witness graph}~$\myH_1$ 
with node set~$P\cup Q$, where the nodes in $P$ correspond to the objects in~$F$ and the nodes in $Q$ 
(with their incident arcs) model the adjacencies in~$\ig(F)$.
Next, we replace each node~$q\in Q$ by a certain gadget
whose ``leaves'' are the neighbors of~$q$, and we triangulate the resulting graph.
We then apply the Planar Separator Theorem to obtain a separator for the resulting graph~$\myH_2$.
Finally, we turn the separator for $\myH_2$ into a clique-based separator for~$\ig(F)$.
Next we explain these steps in detail.

\mypara{Step~1: Creating a witness graph.}
To construct a witness graph for~$\ig(F)$ we use the method of Chen~\etal~\cite{CGP-map-graphs}: 
take a point $p_f$ in the interior of each object $f\in F$, 
and take a \emph{witness point} $q\in \bd f \cap \bd f'$ for each pair of touching objects $f,f'\in F$ 
and add arcs from  $q$ to the points~$p_f$ and ~$p_{f'}$.
\begin{figure}
\begin{center}
\includegraphics{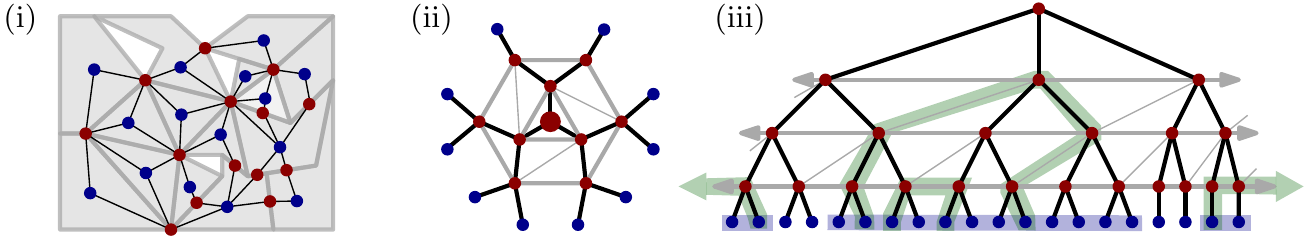}
\end{center}
\caption{(i) A witness graph for the map graph induced by the grey regions. 
    Points in $P$ are blue, points in $Q$ are red. 
(ii) The gadget used to replace a witness point. The edges of~$\tree_q$ are black, 
     the cycles connecting nodes at the same level are grey and thick, 
     the edges to triangulate the 4-cycle are grey and thin. 
(iii) The green paths show an example of how the separator can intersect a gadget. 
      (Note that the tree ``wraps around'', as in part~(ii) of the figure; see also 
      one of the green paths.) The objects added to the 
      clique~$C_q$ correspond to the leaves indicated by the blue rectangles.}
\label{fi:convert-map-graph}
\end{figure}
Let $P=\{p_f : f \in F\}$ and let $Q$ be the set of all witness points added.
We denote the resulting bipartite graph with node set $P\cup Q$ by~$\myH_1$; 
see Figure~\ref{fi:convert-map-graph}(i) for an example.
Observe that points where many objects meet can serve
as witness points for many neighboring pairs in $\ig(F)$.
Chen~\etal~\cite[Lemma 2.3]{CGP-map-graphs} proved that any map graph admits 
a witness set~$Q$ of size~$O(n)$. If the objects in $F$ are polygons with $O(n)$
vertices in total then $Q$ can be found in $O(n)$ time (since the vertices
can serve as the set $Q$.)

\mypara{Step~2: Replacing witness points by gadgets and triangulating.}
We would like to construct a separator for~$\myH_1$ using the Planar Separator Theorem, 
and convert it to a clique-based separator for~$\ig(F)$. For every witness point~$q\in Q$ 
in the separator for~$\myH_1$, the conversion would add a clique $C_q$ to the clique-based separator,
namely, the clique corresponding to all objects $f\in F$ such that $p_f$ is adjacent to~$q$.
However, the node~$q$ adds~1 to the separator size, but the clique~$C_q$ adds~$\log(|C_q|+1)$ to 
the weight of the clique-based separator. To deal with this we modify~$\myH_1$, as follows. 
\medskip

Consider a node~$q\in Q$. Let $N(q)\subseteq P$ denote the set of neighbors of~$q$.
For all nodes $q\in Q$ with $|N(q)|\geq 3$, we
replace the star induced by~$\{q\}\cup N(q)$ by a gadget $G_q$, which is illustrated in
Figure~\ref{fi:convert-map-graph}(ii) and defined as follows.

First, we create a tree $\tree_q$ with root~$q$ and whose leaves are the nodes in~$N(q)$,
as follows. Define the \emph{level} $\ell(v)$ of a node~$v$ in $\tree_q$ to be the distance of $v$
to the root; thus the root has level~0, its children have level~1, and so on.
All leaves in $\tree_q$ are at the same level, denoted~$\ell_{\max}$.
The root has degree~3, nodes at level $\ell$ with $1\leq \ell < \ell_{\max}-1$
have degree~2, and nodes at level $\ell_{\max}-1$ have degree~2 or~1.
For each~$\ell<\ell_{\max}$ we connect the nodes at level~$\ell$ into a cycle. 
After doing so, all faces in the gadget (except the outer face) are triangles or
4-cycles. We finish the construction by adding a diagonal in each 4-cycle.
Define the \emph{height} of a node~$v$ as $\height(v) := \ell_{\max}-\ell(v)$.
The following observation follows from the construction.
\begin{observation}\label{obs:gadget}
Let $v$ be a node at height~$h>0$ in the gadget $G_q$. 
\begin{enumerate}[(i)]
\item The subtree of $\tree_q$ rooted at $v$, denoted $\tree_q(v)$, has at most $3\cdot 2^{h-1}$ leaves.
\item The distance from $v$ to any leaf in $\tree_q$ is at least~$h$.
\end{enumerate}
\end{observation}
To unify the exposition, it will be convenient to also create a gadget for the case where
$q$ has only two neighbors in~$\myH_1$, say $p_f$ and $p_{f'}$. We then define $\tree_q$
to consist of the arcs~$(q,p_f)$ and $(q,p_{f'})$. Note that Observation~\ref{obs:gadget}
holds for this gadget as well.

By replacing each witness point $q\in Q$  with a gadget~$G_q$ as above,
we obtain a (still planar) graph. We triangulate this graph to obtain a maximal
planar graph~$\myH_2$.

\mypara{Step~3: Constructing a separator for~$\myH_2$.}
We now want to apply the Planar Separator Theorem to $\myH_2$. Our final goal
is to obtain a balanced clique-based separator for $\ig(F)$. Hence, we want the separator
for $\myH_2$ to be balanced with respect to~$P$. We will also need the 
separator for $\myH_2$ to be connected. Both properties are guaranteed by the
following version of the Planar Separator Theorem, 
which was proved by Djidjev and Venkatesan~\cite{DV-cycle-separator}.
\begin{quotation}
\vspace*{-2mm} \noindent {\sf\textbf{Planar Separator Theorem.}}
Let $\graph=(V,E)$ be a maximal planar graph with $n$ nodes. Let each node $v\in V$ 
have a non-negative cost, denoted~$\cost(v)$, with $\sum_{v\in V}\cost(v)=1$. 
Then $V$ can be partitioned in $O(n)$ time into three sets $A,B,\sep$ such that (i) $\sep$ is a simple 
cycle of size $O(\sqrt{n})$, (ii) $\graph$ has no arcs between a node in $A$ and a node in $B$,
and (iii) $\sum_{v\in A} \cost(v) \leq 2/3$ and $\sum_{v\in B} \cost(v) \leq 2/3$. \vspace*{-2mm} 
\end{quotation}
When applying the Planar Separator Theorem to $\myH_2$, we set 
$\cost(p) := 1/n$ for all nodes $p_f\in P$ and $\cost(v) := 0$ for all other nodes.
We denote the resulting separator for $\myH_2$ by $\sep(\myH_2)$
and the node sets inside and outside the separator by $A(\myH_2)$ and~$B(\myH_2)$,
respectively.

\mypara{Step~4: Turning the separator for~$\myH_2$ into a clique-based separator for~$\ig(F)$.}
We convert $\sep(\myH_2)$ into a clique-based separator $\sep$ for $\ig(F)$ as follows.
\begin{itemize}
\item For each node $p_f\in\sep(\myH_2)\cap P$ we put the (singleton) clique $\{f\}$ into~$\sep$.
\item For each gadget $G_q$ we proceed as follows. Let $V_q$ be the set of all nodes 
      $v\in \tree_q$ that are in $\sep(\myH_2)$, and define 
      $C_q := \{ f\in F: \mbox{$p_f$ is a leaf of $\tree_q(v)$ that has an ancestor in $V_q$}\}$;
      see Figure~\ref{fi:convert-map-graph}(iii). Observe that $C_q$ is a clique in~$\ig(F)$.
      We add\footnote{We tacitly assume that if an object is in 
      multiple cliques in~$\sep$, we remove all but one of its occurrences.}
      $C_q$ to~$\sep$.
\end{itemize}

The clique-based separator~$\sep$ induces a partition of $F\setminus \bigcup_{C\in\sep}C$
into two parts $A$ and $B$, with  $|A|,|B|\leq 2n/3$, in a natural way,
namely as
$A := \{ f\in F: f\not\in \bigcup_{C\in\sep}C \mbox{ and } p_f \in A(\myH_2)\}$
and
$B := \{ f\in F: f\not\in \bigcup_{C\in\sep}C \mbox{ and } p_f \in B(\myH_2)\}$. The following lemma ensures that $\sep$ is a valid separator.
\begin{lemma}\label{le:map-graph-separator-sep}
There are no arcs in $\ig(F)$ between a node in $A$ and a node in $B$.
\end{lemma}
\begin{proof}
Suppose for a contradiction that there are objects $f\in A$ and $f'\in B$
such that $(f,f')$ is an arc in $\ig(F)$. Let $q\in Q$ be a witness point for the 
arc~$(f,f')$; thus $(p_f,q)$ and $(p_{f'},q)$ are arcs in~$\myH_1$. 
Consider the gadget~$G_q$ and the tree~$\tree_q$. Let $\pi$ denote
the path from $p_f$ to $p_{f'}$ in $\tree_q$. Note that none of the nodes on $\pi$ 
can be in~$\sep(\myH_2)$, otherwise $V_q$ contains an ancestor of $p_f$ or of $p_{f'}$,
and (at least) one of the nodes $p_f,p_{f'}$ is in a clique that was added to $\sep$.
But then the nodes $p_f,p_{f'}$ are still connected in $\myH_2$ after the removal
of~$\sep(\myH_2)$. Hence, we have $p_f,p_{f'} \in A(\myH_2)$ or $p_f,p_{f'} \in B(\myH_2)$,
both contradicting that $f\in A$ and $f'\in B$.
\end{proof}
It remains to prove that $\sep$ has the desired weight.
\begin{lemma}\label{le:map-graph-separator-weight}
The total weight of the separator $\sep$ satisfies $\sum_{C\in \sep} \log(|C|+1)=O(\sqrt{n})$.
\end{lemma}
\begin{proof}
Since $\sep(\myH_2)$ contains $O(\sqrt{n})$ nodes, it suffices to bound the total
weight of the cliques added for the gadgets~$G_q$. 
Consider a gadget~$G_q$. Recall that $V_q$ is the set of all nodes $v\in \tree_q$ that are 
in~$\sep(\myH_2)$. We claim that $\log(|C_q|+1) = O(|V_q|)$, which implies that
$
\sum_q \log(|C_q|+1) = \sum_q O(|V_q|) = O(\sqrt{n}),
$
as desired. It remains to prove the claim. 

Since $\sep(\myH_2)$ is a simple cycle, its intersection with $G_q$ consists of one or more paths. 
Each path~$\pi$ enters and exits $G_q$ at a node in~$N(q)$.
Let $D_{\pi}$ denote the set of all descendants of the nodes in~$\pi$. We will prove 
that $\log(|D_{\pi}|+1) = O(|\pi|)$, where $|\pi|$ denotes the number of nodes of $\pi$.
This implies the claim since
$
\log(|C_q|+1) \leq \sum_{\pi} \log(|D_{\pi}|+1) = \sum_{\pi} O(|\pi|) = O(|V_q|).
$

To prove that $\log(|D_{\pi}|+1) = O(|\pi|)$, let $h_{\max}$ be the maximum height 
of any node in~$\pi$. Thus $|\pi|\geq h_{\max}$ by Observation~\ref{obs:gadget}(ii). 
Consider all subtrees of height~$h_{\max}$ in $\tree_q$. If $\pi$ visits $t$ such subtrees, 
then $|\pi|\geq t$. Moreover, $|D_{\pi}| \leq 3t\cdot 2^{h_{\max}-1}$ by 
Observation~\ref{obs:gadget}(i). Hence,
$
\log(|D_{\pi}|+1) \leq \log \left(3t\cdot 2^{h_{\max}-1}+1\right) < h_{\max} + \log(3t) = O(\max(h_{\max},t)) = O(|\pi|). 
$
\end{proof}
By putting everything together we obtain the following theorem.
\begin{theorem}\label{th:map-graph-separator}
Let $F$ be a set of $n$ interior-disjoint regions in the plane. Then the 
intersection graph $\ig(F)$ has a clique-based balanced separator of weight~$O(\sqrt{n})$.
The separator can be computed in $O(n)$ time, assuming that
the total complexity of the objects in $F$ is~$O(n)$. 
\end{theorem}

\section{Pseudo-disk graphs}
\label{sec:pseudo-disks}
Our clique-based separator construction for a set $F$ of pseudo-disks  uses 
so-called planar supports, defined as follows. Let $\myH$ be a hypergraph with 
node set~$Q$ and hyperedge set~$H$. A graph $\gsup$ is a \emph{planar support}~\cite{RamanRay20} for $\myH$ if 
$\gsup$ is a planar graph with node set $Q$ such that for any hyperedge $h\in H$ the 
subgraph of $\gsup$ induced by the nodes in~$h$ is connected. In our application we let
the node set $Q$ correspond to a set of points stabbing all pairwise intersections between 
the pseudo-disks, that is, for each intersecting pair $f,f'\in F$ there will be a point~$q\in Q$ that
lies in~$f\cap f'$. The goal is to keep the size of $Q$ small, by capturing all intersecting
pairs with few points.
The hyperedges are defined by the regions in~$F$, that is, for every $f\in F$
there is a hyperedge~$h_f := Q\cap f$. Let $\myH_Q(F)$ denote the resulting hypergraph.
\begin{lemma}\label{le:planar-support-to-sep}
Let $F$ be a set of $n$ objects in the plane, let $Q$ be a set of points stabbing all 
pairwise intersections in~$F$, and let $\myH_Q(F)$ denote the hypergraph as defined above.
If $\myH_Q(F)$ has a planar support~$\gsup$ then $\ig(F)$ has a clique-based separator
of size $O(\sqrt{|Q|})$ and weight $O(\sqrt{|Q|} \log n)$.
\end{lemma}
\begin{proof}
Let $\sep(\gsup)$ be a separator for $\gsup$ of size $O(\sqrt{|Q|})$, which exists
by the Planar Separator Theorem, and let $A(\gsup)$ and $B(\gsup)$ be the corresponding separated parts. 
To ensure an appropriately balanced separator we
use the cost-balanced version of the Planar Separator Theorem, as stated in the previous section.
For each object~$f\in F$ we give one point $q_f\in Q\cap f$ a cost of~$1/n$ and all
other points cost~0. We call $q_f$ the \emph{representative} of~$f$.
(We assume for simplicity that each $f\in F$ intersects at least one
other object~$f'\in F$, so we can always find a representative.
Objects $f\in F$ not intersecting any other object are singletons in $\ig(F)$ and can
be ignored.)
For a point $q\in Q$, define $C_q$ to be the clique in~$\ig(F)$
consisting of all objects $f\in F$ that contain~$q$. Our clique-based separator $\sep$ for $\ig(F)$
is now defined as $\sep := \{ C_q : q\in \sep(\gsup) \}$, and the two separated parts are
defined as:
$A := \{ f\in F: f\not\in\sep \mbox{ and } q_f\in A(\gsup)  \}$ and
$B := \{ f\in F: f\not\in\sep \mbox{ and } q_f\in B(\gsup)  \}$. 
Clearly, the size of $\sep$ is $O(\sqrt{|Q|})$ and its weight is $O(\sqrt{|Q|}\log n)$. Moreover, $|A|,|B|\leq 2n/3$
because $\sep(\graph^*)$ is balanced with respect to the node costs.

We claim there are no arcs in $\ig(F)$ between a node in $A$ and a node in $B$.
Suppose for a contradiction that there are intersecting objects~$f,f'$ such that
$f\in A$ and $f'\in B$. By definition of~$Q$ there is a point $q\in Q$ that lies
in~$f\cap f'$. By the planar-support property, the hyperedge~$h_f$
induces a connected subgraph of~$\gsup$, so there is a path $\pi$ that connects
$q$ to the representative~$q_f$ and such that all nodes of~$\pi$ are points in~$f\cap Q$.
No node on the path~$\pi$ can be in~$\sep(\gsup)$, otherwise $f$ is in a clique that was added to~$\sep$.
Similarly, there is a path $\pi'$ connecting~$q_{f'}$ to $q$ such that no point on~$\pi'$
is in~$\sep(\gsup)$. But then there is a path from $q_f$ to $q_{f'}$ in $\gsup$ after the
removal of $\sep(\gsup)$. Hence, $q_f$ and $q_{f'}$ are in the same part of the partition,
which contradicts that $f$ and $f'$ are in different parts.

We conclude that $\sep$ is a clique-based separator with the desired properties.
\end{proof}
\begin{remark}
The witness set~$Q$ in the previous section stabs
all pairwise intersections of objects in the map graph, and so $P\cup Q$
stabs all pairwise intersections as well. $P\cap Q$ has planar support, so we can
get a separator for map graphs using Lemma~\ref{le:planar-support-to-sep}.
Its weight would be $O(\sqrt{n}\log n)$, however, while in the previous section
we managed to get $O(\sqrt{n})$ weight. 
\end{remark}

\mypara{Polygonal pseudo-disks.}
We now apply Lemma~\ref{le:planar-support-to-sep} to obtain a clique-based separator
for a set~$F$ of polygonal pseudo-disks. To this end, let $Q$ be the set of vertices 
of the pseudo-disks in~$F$. Observe that whenever two pseudo-disks intersect, one must have a vertex 
inside the other. Indeed, either one pseudo-disk is entirely inside the other, 
or an edge~$e$ of $f$ intersects an edge~$e'$ of $f'$. In the latter case, one of the two edges ends 
inside the other pseudo-disk, otherwise there are three intersections between the boundaries.
Furthermore, pseudo-disks have the \emph{non-piercing} property:
$f\setminus f'$ is connected for any two pseudo-disks~$f,f'$.
Raman and Ray~\cite{RamanRay20} proved\footnote{Raman and Ray assume the sets $F$ and $Q$
defining the hypergraph are in general position. Therefore we first slightly perturb
the pseudo-disks in $F$ to get them into general position (while keeping the same intersection graph),
then we take $Q$ to be a point set coinciding with the vertex set of $F$, and then we
slightly move the points in $Q$ such that the hypergraph remains the same.}
that the hypergraph~$\myH_Q(F)$ of a set of non-piercing regions has a planar support
for any set $Q$, so in particular for the set~$Q$ just defined.
We can thus apply Lemma~\ref{le:planar-support-to-sep} to compute a clique-based
separator for~$\ig(F)$. The time to compute the separator is dominated by 
the computation of the planar support, which takes $O(n^3)$ time~\cite{RamanRay20}.
\begin{theorem}\label{th:pseudo-square-separator}
Let $F$ be a set of $n$ polygonal pseudo-disks in the plane with $O(n)$ vertices in total. 
Then the intersection graph $\ig(F)$ has a clique-based balanced separator of 
size~$O(\sqrt{n})$ and weight~$O(\sqrt{n}\log n)$, which can be found in $O(n^3)$ time.
\end{theorem}

\mypara{Arbitrary pseudo-disks.}
To construct a clique-based separator using Lemma~\ref{le:planar-support-to-sep}
we need a small point set $Q$ that stabs all pairwise intersections. Unfortunately,
for general pseudo-disks a linear-size set $Q$ that stabs all intersections
need not exist: there is a collection of $n$ disks such that stabbing
all pairwise intersections requires $\Omega(n^{4/3})$ points. (Such a collection can be derived
from a construction with $n$ lines and $n$ points with~$\Omega(n^{4/3})$ incidences~\cite{PS-incidences}.)
Hence, we need some more work before we can apply Lemma~\ref{le:planar-support-to-sep}.
\medskip

Our separator result for arbitrary pseudo-disks works in a more general setting, namely 
for sets from a family $\F$ with linear union complexity. (We say that $\F$ has 
union complexity $U(n)$ if, for any $n\geq 1$ and any subset $F\subset \F$ of size $n$, 
the union complexity of $F$ is~$U(n)$.) Recall that the union complexity of a family
of pseudo-disks is $O(n)$~\cite{KLPS-pseudo-disk-union}.
The next theorem states that such sets admit 
a clique-based separator of sublinear weight. Note that the bound only
depends on the number of objects, not on their complexity.
\begin{theorem}\label{th:linear-union-separator}
Let $F$ be a set of $n$ objects from a family $\F$ of union complexity~$U(n)$, where $U(n)\geq n$.
Then $\ig(F)$ has a clique-based separator of size~$O((U(n))^{2/3})$ and weight~$O((U(n))^{2/3} \log n)$.
In particular, if $F$ is a set of pseudo-disks then $\ig(F)$ has a clique-based separator of 
size~$O(n^{2/3})$ and weight~$O(n^{2/3} \log n)$. The separator can be computed in
$O(n^3)$ time, assuming the total complexity of the objects is~$O(n)$.
\end{theorem}
\begin{proof}
We construct the separator~$\sep$ in two steps.

The first step proceeds as follows. For a point~$p$ in the plane, let $C_p$ denote
the set of objects from the (current) set~$F$ containing~$p$. As long as there is a point~$p$ such that
$|C_p|> n^{1/3}$, we remove $C_p$ from $F$ and put $C_p$ into~$\sep$;
here $n$ refers to the size of the initial set~$F$. Thus the first step
adds $O(n^{2/3})$ cliques to~$\sep$ with total weight~$O(n^{2/3}\log n)$.
This step can easily be implemented in $O(n^3)$ time.

In the second step we have a set $F^*\subseteq F$ of $n^*$ objects with ply~$k$, 
where $n^*\leq n$ and~$k\leq n^{1/3}$. Let $\A(F^*)$ denote the arrangement induced by~$F^*$. 
Since $F^*$ has ply~$k$, the Clarkson-Shor technique~\cite{CW-random-sampling-II}
implies that the complexity of the arrangement~$\A(F^*)$ is~$O(k^2\cdot U(n^*/k))$.
We can compute this arrangement in $O(k^2\cdot U(n^*/k)\log n) = O(n^{2}\log n)$ time~\cite{bcko-cgaa-08}.
Take a point $q$ in each face of the arrangement, and let $Q$ be the resulting set of $O(k^2\cdot U(n^*/k))$ points.
The set~$Q$ stabs all pairwise intersections and the dual graph $\graph^*$ of the arrangement~$\A(F^*)$
is a planar support for the hypergraph~$\myH_Q(F)$. 
Hence, by Lemma~\ref{le:planar-support-to-sep}
there is a clique-based separator $\sep^*$ for $\ig(F)$ of size~$O\left(k\sqrt{U\left(n^*/k\right)}\right)$
and weight~$O\left(k\sqrt{U\left(n^*/k\right)}\log{n^*}\right)$. Note that $U(n)$ is a superadditive function \cite{union-complexity} which implies that $U(n/k)\leq U(n)/k$ and therefore $k\sqrt{U\left(n^*/k\right)} \leq \sqrt{k ~ U(n)}\leq (U(n))^{2/3}$.
By adding $\sep^*$ to the set~$\sep$ of cliques generated
in the first step, we obtain a clique-based separator with the desired properties.
\end{proof}

\section{Geodesic disks inside a simple polygon}
\label{sec:geodesic-disks}
Let $P$ be a simple polygon. 
We denote the shortest path (or: \emph{geodesic}) in $P$ between two points $p,q\in P$ by $\pi(p,q)$;
note that $\pi(p,q)$ is unique since $P$ is simple.
The \emph{geodesic distance} between $p$ and $q$ is defined to be $\|\pi(p,q)\|$,
where $\|\pi\|$ denotes the Euclidean length of a path~$\pi$.
For a given point~$q\in P$ and radius $r>0$, we call the region $D(q,r) := \{p\in P : \|\pi(p,q)\| \leq r \}$
a \emph{geodesic disk}. 
Let $\D=\{D_1,\ldots,D_n\}$ be a set of geodesic disks in~$P$. 
To construct a clique-based separator for $\ig(\D)$ 
we will show that $\D$ behaves as a set of pseudo-disks so we can apply the result of the previous section.

\mypara{The structure of a geodesic disk.}
The boundary $\bd D(q,r)$ of a geodesic disk $D(q,r)$ consists of circular arcs lying 
in the interior of $P$ (centered at $q$ or at a reflex vertex of $P$) and parts of the edges of $P$. 
We split $\bd D(q,r)$ into \emph{boundary pieces} 
at the points where the circular arcs meet~$\bd P$. This generates two sets of boundary pieces: 
a set containing the pieces that consist of circular arcs, and a set~$\Gamma(D)$ 
containing the pieces that consist of parts of edges of~$P$. 
An example can be seen in Fig.~\ref{fi:geodesic-disk}.
\begin{figure}
\begin{center}
\includegraphics{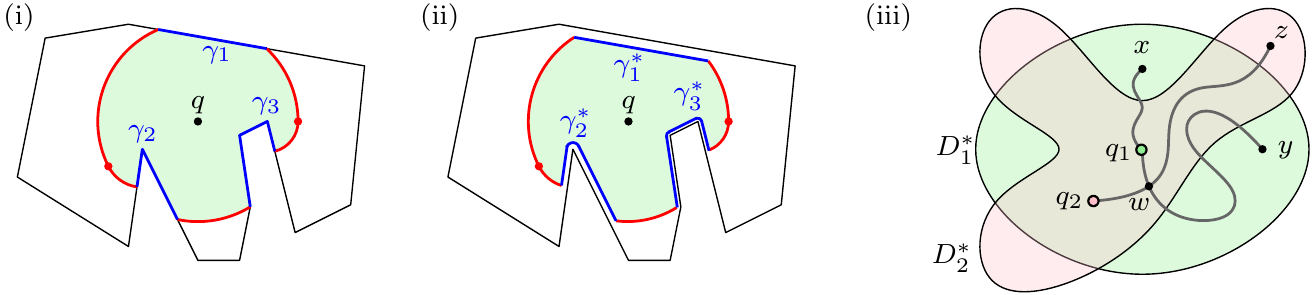}
\end{center}
\caption{(i) A geodesic disk $D$ with center~$q$ and radius~$r$. 
             The set~$\Gamma(D)$ has three pieces, $\gamma_1,\gamma_2$ and $\gamma_3$, shown in blue. 
        (ii) The result of the perturbation. Note that $|\gamma_1|< |\gamma_2|<|\gamma_3|$ 
             and so $\eps_{\gamma_1}>\eps_{\gamma_2}>\eps_{\gamma_3}$.
         (iii) Illustration for the proof of Theorem~\ref{th:geodesic-disks-are-pseudo-disks}.}
\label{fi:geodesic-disk}
\end{figure}

A region~$R\subseteq P$ is \emph{geodesically convex} if for any points $p,q\in R$ 
we have $\pi(p,q) \subseteq R$. Pollack~\etal~\cite{Poll89} showed that geodesic disks inside a 
simple polygon are geodesically convex. An immediate consequence is that the intersection of two geodesic disks is connected. 

\mypara{Geodesic disks behave as pseudo-disks.}
Geodesic disks in a simple polygon are not proper pseudo-disks. For example, if $D_1$ and
$D_2$ are the blue and pink pseudo-disk in the third image in Fig.~\ref{fi:int-graphs},
then $D_1\setminus D_2$ has two components, which is not allowed for pseudo-disks. 
Nevertheless, we will show that $\D$ behaves as a set of pseudo-disks in the
sense that a small perturbation turns them into pseudo-disks, while keeping the intersection
graph the same.

As a first step in the perturbation, we increase the radius of each geodesic disk~$D_i\in\D$
by some small~$\eps_i$. We pick these $\eps_i$ such that the intersection graph~$\ig(\D)$ stays the same
while all degeneracies disappear. In particular, the boundary pieces of different geodesic disks
have different lengths after this perturbation, and no two geodesic disks touch.
With a slight abuse of notation, we still denote the resulting set of geodesic disks by~$\D$. 

The second step in the perturbation moves each $\gamma\in \cup_{i=1}^{n}\Gamma(D_i)$
into the interior of the polygon over some distance~$\eps_{\gamma}$,
which is smaller than any of the perturbation distances chosen in the first step. 
More formally, for each $\gamma \in \Gamma(D_i)$ we remove
all points from $D_i$ that are at distance less than $\eps_{\gamma}$ from $\gamma$; see Fig.~\ref{fi:geodesic-disk}(ii). 
To ensure this gives a set of pseudo-disks we choose the perturbation distances~$\eps_{\gamma}$ 
%
according to the reverse order of the Euclidean lengths of the pieces. 
That is, if $\|\gamma\| > \| \gamma'\|$ then we pick $\eps_{\gamma} < \eps_{\gamma'}$.
The crucial property of this scheme is that whenever $\gamma_i \subset \gamma_j$ 
then $\gamma_i$ is moved more than $\gamma_j$. 

We denote the perturbed version of $D_i$ by $D_i^*$ and define $\D^* := \{ D^*_i : D_i\in\D\}$.
The perturbed versions have the following important property:
\begin{lemma}
\label{le:connected}
Let $D^*_i,D^*_j\in\D^*$.
Every connected component of $D_i^* \setminus D_j^*$ contains a point $u$ with $u\in D_i \setminus D_j$.
\end{lemma}
\begin{proof}
The only case of interest is when there exist $\gamma_i \in \Gamma(D_i)$ and 
$\gamma_j \in \Gamma(D_j)$ such that $\gamma_i \cap \gamma_j \neq \emptyset$.
Suppose now there is a connected component of $D_i^* \setminus D_j^*$ which does not contain 
such a point. Since $D_i^*\subset D_i$ and $D_j^* \subset D_j$, this connected component must 
have been introduced due to the perturbation  (in other words its intersection with all 
connected components of $D_i\setminus D_j$ must be empty). This means that this connected component 
must be the region contained between $\gamma_i^*$ and $\gamma_j^*$  and which was 
previously contained in $D_i \cap D_j$. It is enough to only examine the case when $\gamma_i$ is 
bigger than $\gamma_j$ as otherwise this component only contains points of $D_j^*$. 
There are two cases: $\gamma_j\subset \gamma_i$, or $\gamma_j \not \subset \gamma_i$. 
In both cases note that there will exist an endpoint $p$ of $\gamma_j^*$ that lies inside $D_i^*$. 
Choose now a point $q$ in $(D_j\setminus D_j^*) \setminus (D_i\setminus D_i^*)$ at 
distance $\eps$ from $\gamma_j$ with $\eps_{\gamma_i} < \eps < \eps_{\gamma_j}$. 
See Fig.~\ref{fi:perturbation-proof} for an illustration of the case $\gamma_j\subset \gamma_i$. 
Consider the dashed path starting at $q$ that stays at distance~$\eps$ from $\bd P$
going towards~$p$. Due to the definition of our perturbation this path 
will stay inside $D^*_i$ and remain outside $D_j^*$. Hence, after exiting $D^*_j$
near the point $p$ we reach a point $u \in D_i \setminus D_j$, which is
a contradiction.
\end{proof}
\begin{figure}
\begin{center}
\includegraphics{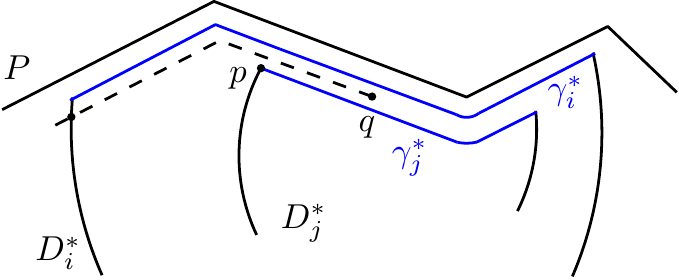}
\end{center}
\caption{Illustration for proof of Lemma \ref{le:connected}.}
\label{fi:perturbation-proof}
\end{figure}
\begin{theorem}\label{th:geodesic-disks-are-pseudo-disks}
Any set $\D$ of geodesic disks inside a simple polygon~$P$ can be slightly perturbed
such that the resulting set $\D^*$ is a set of pseudo-disks with $\ig(\D)=\ig(\D^*)$.
\end{theorem}
\begin{proof}
%
Consider the set $\D^*$ resulting from the perturbation described above.
Suppose for a contradiction that there exist two objects $D_1^*,D_2^*\in \D^*$ such that
$\bd D^*_1$ and $\bd D^*_2$ cross four or more times.
Recall that the intersection of two geodesic disks is connected. This property
is not invalidated by the perturbation. Hence, if $\bd D^*_1$ and $\bd D^*_2$ cross four or more times
then $D^*_1\setminus D^*_2$ (and, similarly, $D^*_2\setminus D^*_1$) has two or more components.

For $i=1,2$, let $q_i$ and $r_i$ denote the center and radius of $D_i$. 
Without loss of generality assume that $r_1\leq r_2$.
Let $x$ and $y$ be points in different components of $D^*_1\setminus D^*_2$; 
see Fig.~\ref{fi:geodesic-disk}~(iii).
By Lemma \ref{le:connected} we can pick $x$ and $y$ such that $x,y \in D_1\setminus D_2$. 
By concatenating the geodesics $\pi(x,q_1)$ and $\pi(q_1,y)$ we obtain a curve that splits
$D_2^*$ into at least two parts---this is independent of where $q_1$ lies, or whether 
$\pi(x,q_1)$ and $\pi(q_1,y)$ partially overlap.
(Note that these geodesics lie in $D_1$ but not necessarily in $D^*_1$. However, they
cannot ``go around'' a component of $D_2^*\setminus D^*_1$, because 
$D_1$ cannot fully contain such a component by Lemma \ref{le:connected}. Hence,
$\pi(x,q_1)\cup\pi(q_1,y)$ must indeed go through $D^*_2$.)
Not all components of $D^*_2\setminus D^*_1$
can belong to the same part, otherwise $x$ and $y$ would not be in different components
of  $D^*_1\setminus D^*_2$. Take a point~$z\in D^*_2\setminus D^*_1$ that lies
in a different part than $q_2$, the center of~$D_2$. 
Again by Lemma \ref{le:connected} we can pick $z$ such that $z \in D_2\setminus D_1$.  
Then the geodesic $\pi(q_2,z)$ must cross $\pi(x,q_1) \cup \pi(q_1,y)$, say at a
point~$w \in \pi(q_1,y)$. Since $z\not\in D_1$ and $y\not\in D_2$ we must have 
$\|\pi(q_1,w)\cup\pi(w,z)\| + \|\pi(q_2,w)\cup\pi(w,y)\| > r_1 + r_2$. But this gives a contradiction 
because $y\in D_1$ and $z\in D_2 $ implies
$\|\pi(q_1,w)\cup\pi(w,y)\|+\|\pi(q_2,w)\cup\pi(w,z)\| \leq r_1+r_2$.

It remains to show that $\ig(\D)=\ig(\D^*)$. As mentioned earlier, the increase of the radii 
in the first step of the perturbation is chosen sufficiently small so that no new intersections are introduced.
The second step shrinks the geodesic disks, so no
new intersections are introduced in that step either. Finally, the fact that the perturbations in the second step
are smaller than in the first step guarantees that no intersections are removed. 
\end{proof}
Theorem~\ref{th:geodesic-disks-are-pseudo-disks} allows us to apply Theorem~\ref{th:linear-union-separator}.
When doing so, we actually do not need to perturb the geodesic disks. We only use the perturbation to argue that the
number of faces in the arrangement defined by $n$ geodesic disks of ply~$k$ is $O(nk)$.
Computing the geodesic disks (and then computing the separator) can be done in polynomial time
in $n$ and the number of vertices of~$P$.
We obtain the following result.
\begin{corollary}\label{col:geodesic-disk-separator}
Let $\D$ be a set of $n$ geodesic disks inside a simple polygon with $m$ vertices.
Then $\ig(\D)$ has a clique-based separator of size~$O(n^{2/3})$ and weight~$O(n^{2/3} \log n)$,
which can be computed in time polynomial in $n$ and $m$.
\end{corollary}
\section{Visibility-restricted unit-disk graphs inside a polygon}
\label{sec:visibility}

Let $P$ be a simple polygon, possibly with holes, and let $Q$ be a set of $n$ points inside~$P$.
We define $\gvis(Q)$ to be the visibility-restricted unit-disk graph of $Q$. The nodes
in $\gvis(Q)$ correspond to the points in~$Q$ and there is an edge
between two points $p,q\in Q$ iff $|pq|\leq 1$ and $p$ and $q$ see each other. 
A vertex of $P$ is \emph{reflex} if its angle within the polygon is more than 180~degrees; note that for a vertex of a hole we look at the angle within $P$, not within the hole. Our separator construction will make use of the notion of a \emph{centerpoint of $P$} which we introduce below.

\subsection{Centerpoint of a simple polygon}
\label{app:centerpoint}
This is a natural extension of the notion of the centerpoint for 
a set $Q$ of $n$ points in the plane, which is a point~$p$ (not necessarily from $Q$) such that
any line through~$p$ divides the plane in two half-planes each containing at most $2n/3$ 
of the points from~$Q$. Recall that a \emph{chord} in $P$ is a line segment $s\subset P$ that connects
two points on the boundary of~$P$. Since we consider $P$ to be a closed set,
a chord may pass through one or more reflex vertices of~$P$.
We say that a chord is \emph{maximal} if it cannot be extended without exiting the polygon.
We define a \emph{half-polygon} of $P$ to be a sub-polygon of $P$ that is bounded by
a (not necessarily maximal) chord~$s$ and a portion $\gamma\subset \bd P$ of the boundary of $P$
such that $s\cap \gamma=\{z_1,z_2\}$, where $z_1,z_2$ are the endpoints of~$s$.
Note that any chord splits $P$ into two or more half-polygons.
\begin{definition}[Centerpoint in a Simple Polygon]
Let $P$ be a simple polygon and $Q$ be a set of $n$ points in $P$. 
A \emph{centerpoint for $Q$ in $P$} is a point $p\in P$ such that any maximal chord through $p$ 
splits $P$ into half-polygons that each contain at most $2n/3$ points from $Q$ in their interior. 
\end{definition}
\begin{theorem} \label{th:centerpoint}
For any simple polygon $P$ and point set $Q\subset P$,  a centerpoint for $Q$ in $P$ exists.
\end{theorem}
\begin{proof}
To prove the theorem we can follow the standard proof~\cite{m-lecture-notes} of the existence of a centerpoint in the plane
almost verbatim.
Thus, consider the family $H$ of all half-polygons of $P$ whose interior contains more than $2n/3$ points of~$Q$. 
Any three of these half-polygons must intersect, because the complement of each of them contains 
less than $n/3$ points from~$Q$. Recall that Helly's Theorem states that if we have a family
of convex sets in the plane such that any three of them intersect, then they all intersect in
a common point. While our half-polygons are not convex, the same result is true here, as follows
from Molnar's Theorem ~\cite{Breen98} which states that if a family of simply connected compact sets 
in the plane is such that any two members have a connected intersection and any three members 
have a non-empty intersection, then the intersection of the family is non-empty. 
Our family $H$ satisfies these conditions and hence $\bigcap H \neq \emptyset$.

We claim that any point $c$ in the common intersection of the half-polygons in $H$ is a centerpoint
for $Q$ in $P$. To see this take a chord~$s$ through $c$ and let $H_1,H_2,\ldots$ be the
half-polygons into which $s$ splits~$P$. Suppose for contradiction that some $H_i$ contains more than $2n/3$ points of $Q$. 
Let $s_i\subseteq s$ be the part of~$s$ bounding~$H_i$.
Then, if we slightly shrink $H_i$ by moving~$s_i$ infinitesimally, we obtain a half-polygon~$H'_i$
that contains more than~$2n/3$ vertices, contradicting the definition of~$c$.
\end{proof}

Our main result is the following:
\begin{theorem}\label{th:vis-separator}
Let $Q$ be a set of $n$ points inside a polygon (possibly with holes) with $r$ reflex vertices.
Then $\gvis(Q)$ admits a clique-based separator of size~$O(\min(n,r)+\sqrt{n})$ and 
weight~$O(\min(n,r\log(n/r)+\sqrt{n}))$.
The bounds on the size and weight of the separator are tight in the worst case,
even for simple polygons.
\end{theorem}

\mypara{The lower bound.}
Recall that even for non-visibility restricted unit-disk graphs, $\Omega(\sqrt{n})$ is a lower
bound on the worst-case size of the separator. Hence, to prove the lower bound of 
Theorem~\ref{th:vis-separator} it suffices to give an example where the size and weight
are $\Omega(\min(n,r))$ and $\Omega(\min(n,r\log (n/r)))$, respectively. This example
is given in Fig.~\ref{fi:vis-lower-bound}(i). 
\begin{figure}
\begin{center}

\includegraphics[clip,width=\columnwidth]{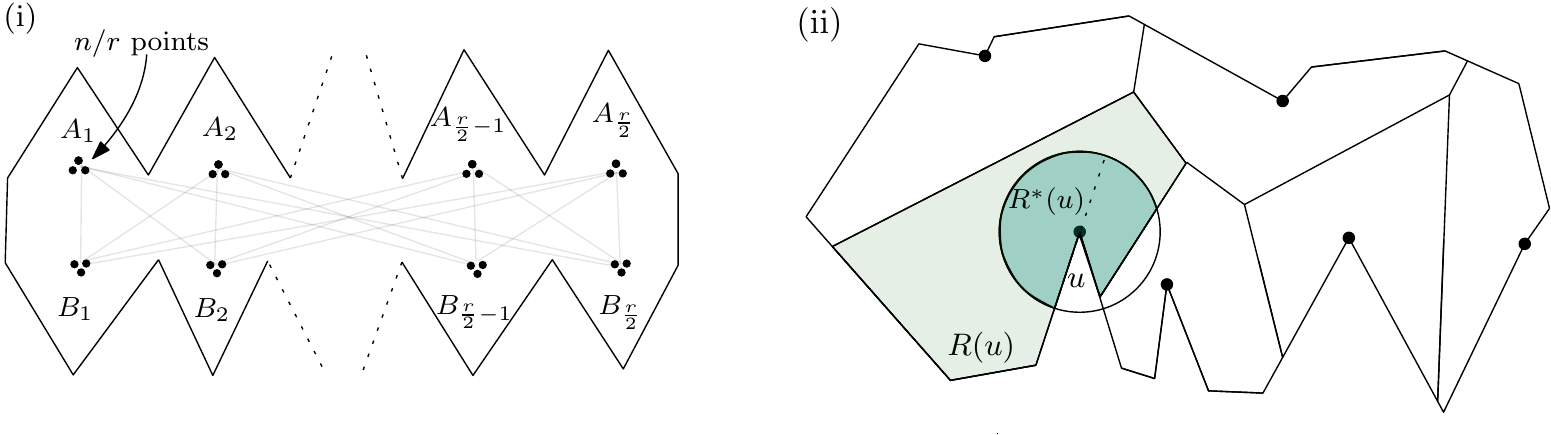}
\end{center}
\caption{(i) Each cluster $A_i$ sees any of the clusters $B_j$ completely and all distances
         are at most~1, so a separator
         that splits $\gvis\left(\bigcup A_i \cup \bigcup B_j\right)$ into two or more
         components must fully contain $\bigcup A_i$ or $\bigcup B_j$. Since the
         clusters~$A_i$ (and similarly $B_j$) do not see each other, such a separator has
         size at least $r/2$ and weight at least~$((r/2)\log(n/r))$.
         (ii) Splitting $R^*(u)$ into two convex parts.}
\label{fi:vis-lower-bound}
\end{figure}

\mypara{The upper bound.}
We first describe our construction for polygons without holes.
Our separator construction has two steps. In the first step we put all points that can see
a reflex vertex within distance~$\sqrt{2}$ into the separator. We will argue that we
can do this in such a way that we put $O(1)$ cliques per reflex vertex into
the separator. In the second step we handle the remaining points. We take a centerpoint~$c$
inside~$P$ and then define $\sqrt{n}$ chords through~$c$. For each chord~$s_i$ we put the
points within distance~1/2 into the separator, suitably grouped into cliques. We will argue that
the total weight of the cliques, over all chords~$s_i$ is $O(n)$.
Hence, there is a chord whose cliques have total weight $O(\sqrt{n})$.
Adding the weight of the cliques we added in Step~1 then gives us the desired separator.

\paragraph*{Step~1: Handling points that see a nearby reflex vertex} 
Let $\Vr$ be the set of reflex vertices of $P$, and let $Q_1\subseteq Q$ be the set 
of points that can see a reflex vertex within distance~$\sqrt{2}$. 
Consider the geodesic Voronoi diagram of $\Vr$ within~$P$. Let $R(u)$ be the Voronoi region of 
vertex $u\in \Vr$ and define $R^*(u) := R(u) \cap D(u,\sqrt 2)$. 
Note that all points in $R(u)$ can see~$u$ and that all points in $Q_1$ are in~$R*(u)$ for
some vertex~$u\in \Vr$. The following lemma is illustrated in Fig.~\ref{fi:vis-lower-bound}(ii).
\begin{lemma} \label{le:split-Rstar}
Let $e$ be an edge incident to $u$. If we extend $e$ until it hits $\bd R^*(u)$
then $R^*(u)$ is split into two convex parts.
\end{lemma}
\begin{proof}
It suffices to argue that the angle of $R(u)$ corresponding to $u$ is the only reflex angle within $R(u)$. 
Then the result will follow from the fact that $D(u,\sqrt 2)$ is convex. Clearly the shortest path connecting 
any point $p \in R(u)$ to $u$ is the segment~$pu$. This implies that $\bd R(u)$ consists of 
straight line segments that are parts of perpendicular bisectors of the segments connecting 
$u$ to other reflex vertices of $P$. Any two such consecutive edges then cannot form a reflex angle within $R(u)$.
\end{proof}
Since $R^*(u)$ has diameter~$O(1)$, Lemma~\ref{le:split-Rstar} implies that $Q_1\cap R^*(u)$ 
can be partitioned into $O(1)$ cliques. We collect all these cliques into a set~$\sep_1$.
Since there are~$r$ reflex vertices, $\sep_1$ consists of $O(r)$ cliques whose total
weight is~$O(r\log(n/r))$.

\paragraph*{Step~2: Handling points that do not see a nearby reflex vertex}
%
%
Our separator $\sep$ consists of the cliques in $\sep_1$ plus a set $\sep_2$
of cliques that are found as follows. 
Let $Q_2 := Q\setminus Q_1$ be the set of points that do not see a reflex vertex within distance~$\sqrt{2}$. 
Let $c$ be a \emph{centerpoint for $Q_2$ inside $P$},
that is, a point such that any chord through~$c$ splits~$P$ into two half-polygons containing 
at most $2|Q_2|/3$ points from $Q$.
Such a point exists by Theorem~\ref{th:centerpoint}.
Let $G$ be a $\sqrt{n}\times\sqrt{n}$ grid centered at~$c$,
where we assume for simplicity that $\sqrt{n}$ is integer. 
For each of the $\sqrt{n}$ points~$g_i$ in the rightmost column
of the grid (even if $g_i\not\in P$), we define a maximal chord~$s_i$ by taking the line~$\ell_i$ through $c$ and $g_i$
and then taking the component of $\ell_i\cap P$ that contains~$c$; see Fig.~\ref{fi:vis-separator-app}(i).
We assume for simplicity that the chords~$s_i$ doe not pass through reflex vertices;
this can be ensured by slightly rotating the grid, if necessary.

For each chord~$s_i$, we define $Q_2(s_i)$ be the set of points $q\in Q_2$
such that there is a point $z\in s_i$ that sees $q$ with $|qz|\leq 1/2$.
\begin{figure}
\begin{center}
\includegraphics{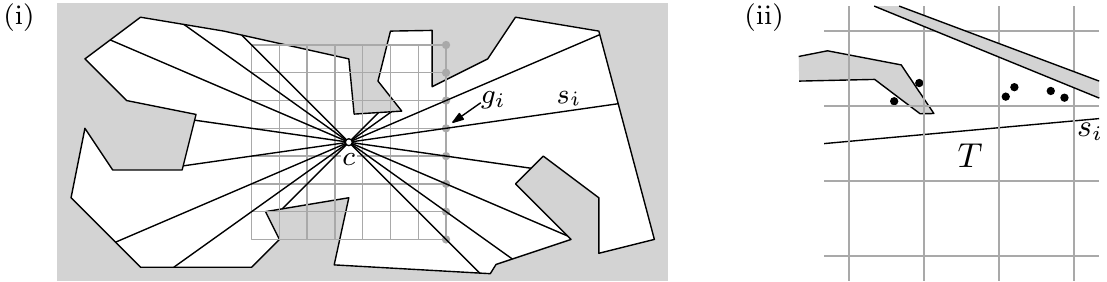}
\end{center}
\caption{(i) The grid $G$ defining the chords~$s_i$. 
          (ii) The points in the top-left cell do not see each other, but they see a reflex vertex
               and so they are not in $Q_2(s_i)$.
               The points in the top-right cell can be partitioned into $O(1)$ cliques.}
\label{fi:vis-separator-app}
\end{figure}
Note that $\gvis(Q)$ cannot have an arc between a point $p\in Q_2 \setminus Q_2(s_i)$ above $s_i$
and a point $q\in Q_2 \setminus Q_2(s_i)$ below $s_i$; otherwise $p$ and/or $q$ see a point on
$z$ within distance~1/2, and so at least one of $p,q$ is in $Q_2(s_i)$. Since $s_i$ is a chord
through the centerpoint~$c$, this means that $s_i$ induces a balanced separator.

It remains to argue that at least one of the chords~$s_i$ induces a separator whose weight
is small enough. We will do this by creating a set~$\sep(s_i)$ of cliques for each chord~$s_i$,
and prove that the total weight of these cliques, over all chords $s_i$, is $O(n)$. Since there
are $\sqrt{n}$ chords, one of them has the desired weight. 

\mypara{Step~2.1: Points outside the grid.}
We simply  put all points from $Q_2(s_i)$ that lie outside the grid $G$---that is, the points
that do not lie inside a grid cell---into $\sep(s_i)$, as singletons.
\begin{observation}\label{obs:outside-grid}
Let $q$ be a point that does not lie in any of the cells of the grid~$G$. Then $q$ lies within
distance~$1/2$ of at most two chords~$s_i$.
\end{observation}
%
%
%
%
%
%
%
%
Observation~\ref{obs:outside-grid} implies that the total number of singleton cliques
of points outside grid cells, summed over all chords~$\sep(s_i)$, is~$O(n)$.

\mypara{Step ~2.2: Points inside the grid.}
Next, consider a cell~$T$ of the grid $G$. Suppose a point $q\in Q_2(s_i)$ sees a point $z\in T\cap s_i$ 
such that $|qz|\leq 1/2$. Then $q$ must lie inside one of the at most nine grid cells surrounding and including~$T$.
Consider such a cell $T'$. 
\begin{lemma}\label{le:inside-grid}
The points from  $Q_2(s_i)\cap T'$ that can see a point on $s_i\cap T$ within distance $1/2$ can be partitioned
into $O(1)$ cliques.
\end{lemma}
\begin{proof}
We will show that we can split these points into at most four cliques. 
To this end we split $T'$ in four equal-sized subcells; see Fig.~\ref{fi:inside-grid}. 
Clearly each subcell has a diameter~$\frac{1}{2}\sqrt{2}<1$. It suffices now to show that if two points 
$p,q$ in the same subcell do not see each other, then they see a reflex vertex within distance $\sqrt 2$; 
hence they are in $Q_1$ and not in $Q_2(s_i)\cap T'$.
\begin{figure}
\begin{center}
\includegraphics{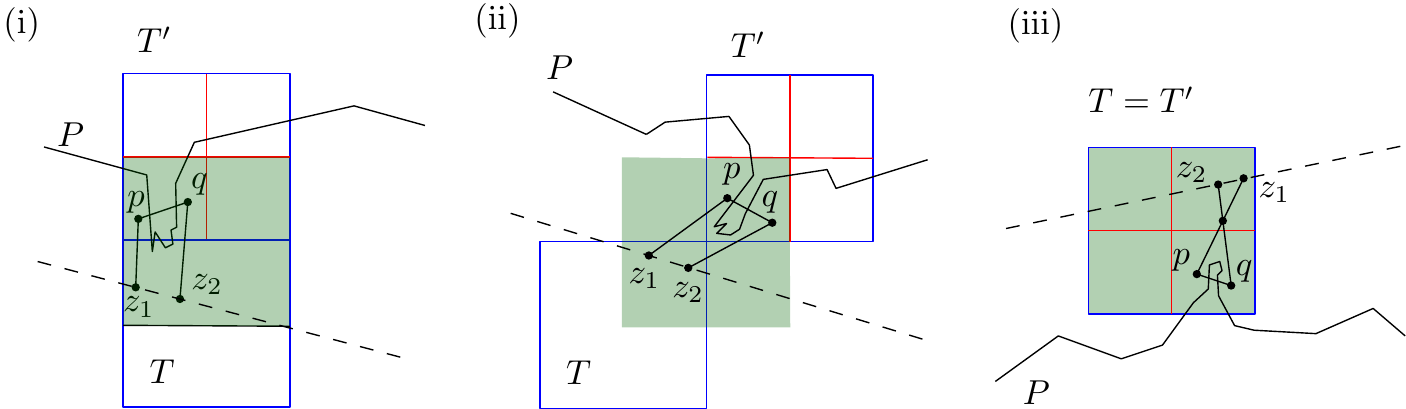}
\end{center}
\caption{Illustration for the proof of Lemma~\ref{le:inside-grid}.}
\label{fi:inside-grid}
\end{figure}

Let $z_1$ and $z_2$ be points on $s_i \cap T$ that see $p,q$, respectively, with
$|pz_1|\leq 1/2$ and $|qz_2| \leq 1/2$. The quadrilateral formed by the points $p,q,z_1,z_2$ must always lie 
within a square of side length~1. This square is shown in green in Fig.~\ref{fi:inside-grid}, for three different  
cases depending on the relative position of~$T'$ and~$T$. In this quadrilateral, the side~$pq$ has to be crossed 
by the $\bd P$, since we assumed $p$ and $q$ do not see each other. Moreover,
$\bd P$ cannot cross $s_i$ or the segments $pz_1,qz_2$. Thus at least one reflex vertex of $\bd P$
lies within the quadrilateral. It's then clear that both $p,q$ will be able to see a reflex vertex 
within the quadrilateral (possibly not the same); since all points lie within a square of side length~1, 
the result follows. Note that in Fig.~\ref{fi:inside-grid} (iii) we have made the edges $pz_1,qz_2$ 
cross each other to showcase that the proof still goes through in this case.
\end{proof}
We thus create $O(1)$ cliques for each cell~$T'$ that is one of the at most nine cells
surrounding a cell $T$ crossed by~$s_i$, and put them into~$\sep(s_i)$.
This adds at most $O(\log (n_{T'}+1))$ weight to $\sep(s_i)$, where $n_{T'} := |Q_2\cap T'|$.
Since there are $n$ cells in total, this immediately gives a total weight of
$O(n\log n)$ over all sets $\sep(s_i)$. Next we show that the the total weight of the cliques is actually $O(n)$. 
In what follows we denote by $\weight(\sep(s_i))$ the total weight of the cliques corresponding to chord $s_i$.
\begin{lemma}
The total weight of the cliques corresponding to the chords $s_i$ is $O(n)$, that is:
\[
\sum_{i=1}^{\sqrt{n}}\weight(\sep(s_i))=O(n).
\]
\end{lemma}
\begin{proof}
Let $\T(s_i)$ denote the cells crossed by a chord~$s_i$, and let $N[T]$ denote the at most
nine cells surrounding a cell~$T$, including $T$ itself.
By the arguments above, the total weight of all separators is.
\begin{equation}\label{eq:sum1}
\sum_{i=1}^{\sqrt{n}} \sum_{T\in\T(s_i)} \sum_{T'\in N[T]} O(\log (n_{T'}+1))
\end{equation}
Let $K_1$ be the set of cells in the central column of the grid, that is, the
column containing the centerpoint~$c$. For $2\leq j\leq m$, where $m:=\lceil \sqrt{n}/2\rceil,$
let $K_j$ be the cells in the pair of columns at distance~$j-1$ from~$K_1$.
For example, $K_2$ are the cells in the columns immediately to the left and right of $K_1$.
Then we can make the following observations.
\begin{itemize}
\item The set~$K_j$ contains $4j-2=O(j)$ cells that are crossed by at least one chord~$s_j$.
\item Cells in $K_j$ are crossed by $O(\sqrt{n}/j)$ chords.
\end{itemize}
Now define $K^*_j$ to be the set of cells $T'$ in $K_j$ such that one of the nine cells
surrounding $T'$ is crossed by at least one chord. Note that each $T'\in K^*_j$
is in $N[T]$ for at most nine cells~$T$, which are all in $K_{j-1}\cup K_j\cup K_{j+1}$.
(Here we define $K_{-1}=K_{m+1}=\emptyset$.)
It then follows from the above observations that $|K_j^*|=O(j)$. Moreover, if
$T'\in K^*_j$ and $T'\in N[T]$, then $T$ is crossed by~$O(\sqrt{n}/j)$ chords.
Hence, we can rewrite~(\ref{eq:sum1}) as
\begin{equation}\label{eq:sum2}
\sum_{j=1}^{m} \left( O\left(\frac{\sqrt{n}}{j}\right) \cdot \sum_{T'\in K^*_j} O(\log (n_{T'}+1)) \right)
=
O\left( \sum_{j=1}^{m}  \frac{\sqrt{n}}{j} \cdot \sum_{T'\in K^*_j} \log (n_{T'}+1) \right)
\end{equation}
Now we are ready to compute the total weight of the separator. Note that
\[
\sum_{T'\in K^*_j} \log (n_{T'}+1) =   \log \left( \prod_{T'\in K^*_j} (n_{T'}+1)) \right)
\]
Now define $n_j$ to be the total number of points in the cells in $K^*_j$.
The AM-GM inequality gives us:
\[
\log \left( {\prod_{T' \in K^*_j }(n_{T'}+1)} \right)
\leq \log \left( \left(\frac{\sum_{T'\in K^*_j} (n_{T'}+1)}{|K_j^*|}\right)^{|K_j^*|} \right)
= |K_j^*| \log \left( \frac{n_j}{|K^*_j|} + 1 \right),
\]
Since $|K^*_j| = O(j)$ we have 
\[
|K_j^*| \log \left( \frac{n_j}{|K^*_j|} + 1 \right) =O\left( j\cdot \log \left(\frac{n_j}{j}+1\right)\right)
\]
Combining this with the previous we conclude that the total weight
over all separators is bounded by
\begin{equation} \label{eq:sum3}
O\left( \sum_{j=1}^{m}  \frac{\sqrt{n}}{j} \cdot j \log \left( \frac{n_j}{j} + 1 \right) \right)
=
O\left( \sum_{j=1}^{m}  \sqrt{n} \cdot \log \left( \frac{n_j}{j} + 1 \right) \right)
\end{equation}
We have that: 
\begin{equation}\label{eq:sum4}
\sum_{j=1}^{m} \log \left( \frac{n_j}{j} + 1 \right) 
=
O\left(\log \left( \frac{\prod_{j=1}^{m} n_{j}}{ m !}\right)\right)
\end{equation}
Recall that $m=\lceil \sqrt{n}/2 \rceil$, which implies that $n/m \leq 2m$.
Furthermore, Stirling's approximation tells us that $\log (m!) = m\log m - m + O(\log m)$.
Since $\sum_{j=1}^m n_j \leq n$ and using the AM-GM inequality we derive
\[
\begin{array}{lll}
\log \left( \frac{\prod_{j=1}^{m} n_{j}}{ m !} \right) &
\leq &
\log \left( (n/m)^m \right) - \log (m!) \\
& = & m \log (n/m) -  \left( m\log m - m + O(\log m) \right) \\
& \leq & m \log (2m) - \left( m\log m - m + O(\log m) \right) \\
& = & 2m + O(\log m) \\
& = & O(\sqrt{n})
\end{array}
\]
Combing this with Equations~(\ref{eq:sum3}) and~(\ref{eq:sum4})
we conclude that the total weight of the sets $\sep(s)$ is $O(n)$, as claimed.
\end{proof}
We conclude that there will be a chord $s_i$ whose cliques have total weight $O(\sqrt n)$. Then our desired separator is the set $\sep=\sep_1 \cup \sep(s_i)$.

\subsection{Extension to polygons with holes}
We now show how to extend the previous approach so that it also works for a polygon~$P$ 
with holes. The key ideas remain the same. 
\medskip

As previously, the first step is to put the points that see a reflex vertex within distance $\sqrt 2$
into the separator. This can again be done using the geodesic Voronoi diagram of the reflex vertices. 
Recall that a vertex of a hole is defined to be reflex if its angle within $P$ (not: within the hole) 
is more than 180~degrees.
\medskip

In the second step we again construct $\sqrt{n}$ potential separators. When $P$ has holes, a chord
may not split $P$ into two half-polygons. Hence, instead of taking chords though a centerpoint 
for~$Q_2$ in~$P$, we proceed slightly differently. We take
a regular centerpoint for~$Q_2$, and for each point~$g_i$ in the rightmost column of
the grid~$G$ we take the full line~$\ell_i$ through $c$ and~$p_i$. Note that $\ell_i\cap P$ 
can consist of many chords. Together, these chords split $P$ into several parts; the parts 
above~$\ell_i$ contain at most $2|Q_2|/3$ points, and the parts below $\ell_i$ contain at most $2|Q_2|/3$ points as well.
Hence, the resulting separators will be balanced.

Points in $Q_2$ that lie outside the grid~$G$ can be handled as before. They are added as 
a singleton to $O(1)$ lines~$\ell_i$, so they contribute weight $O(n)$ in total.

\begin{figure}
\begin{center}
\includegraphics{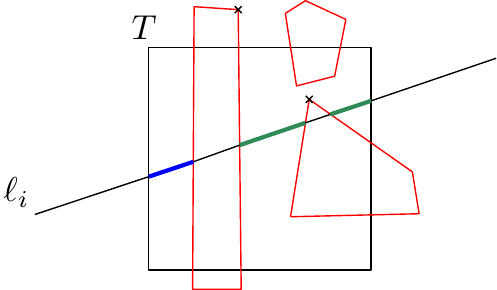}
\end{center}
\caption{A line $\ell_i$ which crosses a cell $T$ that has three holes, colored in red. The blue piece of $\ell_i$ is the entrance piece, while the two green pieces are the non-entrance pieces. Note that each green piece can be uniquely associated to the pair of entrance-exit points of the hole preceding it. Marked with a cross are two possible choices of reflex vertices that are then associated with the green pieces.}
\label{fi:holes}
\end{figure}
We now proceed to describe how to split points in $Q_2$ within the grid into cliques.  
For each line $\ell_i$ and each grid cell $T$ we define an \emph{entrance piece} to be 
the first piece of $\ell_i \cap T \cap P$. The remaining pieces will be referred to as 
\emph{non-entrance pieces} (see Fig.~\ref{fi:holes}). Entrance pieces can be handled 
exactly as in the case without holes: we look at all the points in the nine cells 
surrounding and including $T$ which see this piece within a distance $1/2$ and split 
them in a constant number of cliques. Then the total sum of the weights of these 
cliques for all the lines will be $O(n)$ as before. Thus we can select a line~$\ell_i$ 
such that the weight of the cliques we add for their
for the entrance pieces is $O(\sqrt{n})$. 

It remains now to handle the non-entrance pieces. Let~$\ell_i$ be the line we
selected. We will show that $\ell_i$ 
has $O(r)$ non-entrance pieces and that for each of them we only need to add $O(1)$ cliques. 
For the first part, observe that if a line $\ell_i$ intersects a hole $2k$ times then 
the hole must have at least $k$ reflex vertices of $P$. These $2k$ points of intersection 
can be split into $k$ pairs of entrance-exit points. Now note that a non-entrance piece can be 
uniquely associated with a unique pair of entrance-exit points. Therefore we have associated 
each non-entrance piece with a unique reflex vertex of $P$ and as a result the number of these pieces is indeed $O(r)$.
The second part is handled as before by adding a constant number of cliques from the nine cells that surround the cell that contains each non-entrance piece.
Therefore the cliques that correspond to these pieces add a total weight of $O(r \log \frac{n}{r})$. 

To summarize, after having found a separator such that the total weight of the entrance pieces is $O(\sqrt{n})$
we can add the cliques for the  non-entrance pieces to obtain a separator of weight $O(r \log \frac{n}{r}+\sqrt n)$.
Together with the cliques from $\sep_1$ this gives the final separator.
This finishes the proof of Theorem~\ref{th:vis-separator}.

\section{Applications}\label{sec:applications}
In this section, we show how one can use clique-based separators to get subexponential
algorithms for several classic graph problems, namely \mis, \fvs, and $q$-\col for
constant $q$. The algorithms are very generic: they only use the fact that the
clique-based separators exist and can be computed in polynomial time, so they can 
immediately be combined with any of our separator theorems. 

There are subexponential algorithms for \mis, \textsc{Feedback Vertex Set}, 
and $3$-\col in the class of string graphs. In particular, Bonnet and Rzazewski~\cite{BonnetR19} obtain a running 
time of $2^{O(n^{2/3}\log n)}$ for \mis and $3$-\col, and $2^{O(n^{2/3}\log^{O(1)} n)}$ for \fvs. 
They also show that there is no subexponential algorithm for $q$-\col in string graphs in case of $q\geq 4$. 
All graph classes studied in this paper are a subclass of string graphs, therefore the algorithmic 
results of Bonnet and Rzazewski  carry  over.
In the case of map graphs, one can derive $2^{O(\sqrt{n}\log n)}$ algorithms for \mis and \fvs 
by using the trivial bound $k\leq n$ in the subexponential parameterized algorithms of
Fomin~\etal~\cite{FominLP0Z19}.
Moreover, the separator theorem of Matou\v{s}ek~\cite{C-map-graph-mis} can be used to directly get 
a $2^{O(\sqrt{n})}$ algorithm for $q$-\col for fixed constant $q$.
Finally, for polygons of total complexity $O(n)$, one can use the parameterized algorithm of
Marx and Philipczuk~\cite{MarxP15} to derive a $2^{O(\sqrt{n}\log n)}$ algorithm for \mis.

Using our separator theorems, we get subexponential algorithms for the three problems mentioned above in 
the various graph classes under consideration. The running times we obtain match or slightly
improve the results that can be obtained by applying the existing results mentioned above.
It should be kept in mind, of course, that the existing results are for more general graph classes.
An exception are our results on map graphs, which were explicitly studied before and where we
improve the running time for \mis and \fvs from $2^{O(\sqrt{n}\log n)}$ to $2^{O(\sqrt{n})}$.
(Admittedly, the existing results apply in the parameterized setting while ours don't.)
In any case, the main advantage of our approach is that it allows us to solve \mis, \fvs and
$q$-\col on each of the mentioned graph classes in a uniform manner. Specifically, our result are as follows.

In map graphs that are given by a linear-sized representation we obtain $2^{O(\sqrt{n})}$ algorithms 
for all three problems. In pseudodisk graphs and in intersection graphs of objects with linear union complexity, we get $2^{O(n^{2/3}\log n)}$ algorithms for \mis and \fvs, and a $2^{O(n^{2/3})}$ algorithm for $q$-\col with constant $q$. In intersection graphs of polygonal pseudodisks of total complexity $O(n)$, the obtained running times are $2^{O(\sqrt{n}\log n)}$ for \mis and \fvs, and $2^{O(\sqrt{n})}$ for $q$-\col with constant $q$. For intersection graphs of geodesic disks in a simple polygon, our separators yield the same running times as for pseudodisks: $2^{O(n^{2/3}\log n)}$ algorithms for \mis and \fvs, and a $2^{O(n^{2/3})}$ algorithm for $q$-\col with constant $q$.

We also get subexponential algorithms for
each of our problems in visibility-restricted unit-disk graphs in polygons 
with $r$~reflex vertices (assuming the total number of vertices is polynomial in~$n$). 
The running time is $2^{f(n,r)}$, where $f(n,r)=O(\min(n,r \log(n/r)) +\sqrt{n})$.

\mypara{The algorithms.}
Recall that the \mis problem is to find a maximum-size subset of pairwise non-adjacent nodes in a given graph~$G$. 
The \fvs problem asks 
to find a minimum-size subset of nodes whose deletion makes $G$ cycle-free. 
Here it will be more convenient to work with the complement of the solution, which corresponds to the \textsc{Maximum Induced Forest} problem: find a maximum-size
subset of nodes that induce a forest in $G$. Finally, in the $q$-\col problem,
the goal is to decide if $G$ has a proper $q$-coloring (that is, a coloring of the nodes
using $q$ colors such that adjacent nodes have different colors).
For $q$-\col, we can immediately reject any instance that contains a clique of size at
least $q+1$, as such a clique has no proper $q$-coloring. Therefore the weight of each
clique is at most $\log(q+1)=O(1)$, and we only need to deal with instances of maximum
clique size $q$. Consequently, if we have a clique-based separator $\myS$ with $s$ 
cliques in the graph, then the weight of the separator is at most $s\cdot \log(q+1)=O(s)$.

We can simplify the presentation by abstracting away the geometry. In what follows, we fix a graph class $\G$, and we assume that for any $n$-node graph $G\in \G$ we can compute a balanced clique-based separator $\myS$ of $G$ in polynomial time of size $s_\G(n)$ and weight $w_\G(n)$. Note that our separator computations rely on having a geometric representation of the input graph~$G$.
We will prove the following theorem.


\begin{theorem}
Let $\G$ be a class of geometric intersection graphs, and suppose that $\G$ has a clique-based separator 
theorem of weight $w_\G(n)$ and size $s_\G(n)=n^{c}$ for some constant $0<c<1$. Suppose moreover that these separators can be computed 
in polynomial time.  Then for an $n$-node graph $G\in \G$ given by its geometric representation, one can solve:
\begin{enumerate}[(1)]
    \item \mis in $2^{O(w_\G(n))}$ time,
    \item \fvs in $2^{O(w_\G(n))}$ time,
    \item and $q$-\col for fixed constant $q$ in $2^{O(s_\G(n))}$ time.
\end{enumerate}
\end{theorem}


\begin{proof}
{\sf\textbf{(1)}} We use the same algorithm as De~Berg~\etal~\cite{bbkmz-ethf-20}, which we recall
here. Let $X$ be a maximum independent set. Notice that any clique $C$ can contain at
most one node from~$X$. Therefore, the number of ways that a given set $X$ can intersect
a clique-based separator~$\myS$ is at most
\[
\prod_{C\in \myS} (|C|+1)=2^{\sum_{C\in \myS} \log(|C|+1)}=2^{\weight(\myS)}.
\]
We can use a simple divide-and-conquer strategy to find the set $X$. We compute a
separator~$\myS$, and for all independent sets $X_\myS$ in $\bigcup_{C\in \myS} C$, 
we will find the largest independent set that $X_\myS$ is a part of. 
To this end we consider the graph $G_{X_\myS} := G[V\setminus (N(X_\myS) \cup \bigcup_{C\in \myS} C)]$, that is,
the graph~$G$ where we remove the nodes of the separator as well as those nodes that
neighbor a node in $X_S$. We compute a maximum independent set on each side of the
separator, restricting ourselves to $G_{X_\myS}$, in a recursive fashion. 
The correctness follows from the fact that the set $X^*_\myS:=X\cap (\bigcup_{C\in \myS} C)$ 
will be considered by the algorithm, and then the remainder of $X$ on each side of the
separator will form maximum independent sets within that side.

Since the number of sets $X_S$ to enumerate is at most $2^{\weight(\myS)}\leq 2^{w_\G(n)}$, we get the following recursion for the running time:
\[T(n)=2^{w_\G(n)}\big(2T(\beta n) + \poly(n)\big),\]
where $\beta <1$ is the balance factor of the separator theorem. Solving this recursion yields that $T(n)=2^{O(w_G(n))}$, as required.
\medskip

\noindent {\sf\textbf{(2)}} We solve the complement problem, \textsc{Maximum Induced Forest}. Observe that an
induced forest can have at most two nodes in each clique, therefore an induced forest
with node set $X$ can intersect a clique-based separator $\myS$ in at most
\[
\prod_{C\in \myS} \left(\binom{|C|}{2}+|C|+1\right)< \prod_{C\in \myS}(|C|+1)^2=2^{\sum_{C\in \myS} 2\log(|C|+1)}=2^{2\cdot \weight(\myS)}
\]
many ways. However, unlike with \mis, the solutions on the two sides of the separator are
not completely independent: we must ensure that their union is a forest, and does not
induce any cycles. Let $B$ be a set of nodes in $G$, called \emph{boundary nodes}, and
let $\pi$ be a partition of $B$. We say that a forest $X$ \emph{realizes} $\pi$ if
$B\subseteq X$ and $x,y\in B$ are in the same connected component of $G[X]$ if and only
if they are in the same partition class of $\pi$.

A separator can contain up to $O(|\myS|)$ nodes of the solution forest, which will become boundary nodes in $A\cup S$ and $B\cup S$ with some specific partitions. Formally, we need to solve the following problem: given a graph $G$, a boundary node set $B$ and a partition $\pi$, find the maximum induced forest of $G$ that realizes $\pi$ on $B$. In order to keep the running time small, we need to ensure that this boundary set stays small compared to the size of the instance in the recursion tree. We observe that all of our separators can be extended to achieve balance with respect to any node set $W\subset V(G)$. Using such separators, we can ensure that the number of boundary nodes for an instance of size $n$ during the recursion is at most $O(s_\G(n))$ by for example taking separators that are balanced with respect to all nodes and the boundary set $B$ alternatingly. See~\cite{BergBKK18} for a different way to keep the boundary set small.
As $B$ has $|B|^{O(|B|)}=2^{O(|B|\log|B|)}=2^{O(s_{\G}(n)\log n)}$ partitions, the running time recursion is
\[T(n)=2^{2w_\G(n)+O(s_G(n)\log n)}(T(\beta'n)+\poly(n)),\]
where $\beta'<1$ is a constant that depends only on the separator balance factor $\beta$ and the size function $s_\G$. This gives a running time of \[T(n)=2^{O(w_\G(n)+s_G(n)\log n)}.\]
The running time can be further reduced by using the rank-based approach of Bodlaender~\etal~\cite{rankbased}. Here instead of trying all partitions, one only needs to keep track of a so-called \emph{representative set} of $2^{O(|B|)}$ weighted partitions, and update these representative sets bottom-up in the recursion tree. We can implement this technique similarly to earlier implementations. See~\cite{bbkmz-ethf-20} for an example with maximum induced forest and~\cite{BergBKK18} for using the technique in a separator-based divide-and-conquer algorithm) without a tree-decomposition. The resulting algorithm has the following running time recursion:
\[T(n)=2^{2w_\G(n)+O(s_G(n))}(T(\beta'n)+\poly(n))=2^{O(w_\G(n))}(T(\beta'n)+\poly(n)).\]
Solving the recursion gives the desired running time of $T(n)=2^{O(w_\G(n))}$.
\medskip

\noindent {\sf\textbf{(3)}} We start by checking if there is a clique of size at least $q+1$. 
This can be done by brute force in $O(n^{q+1})=\poly(n)$ time. If such a clique is found
then we can reject. Otherwise, we use a divide-and-conquer strategy: for all proper
$q$-colorings of the separator node set $S=\bigcup_{C\in \myS} C$, we solve the problem
recursively for the node sets $S_1=S\cup A$ and  $S_2=S\cup B$, where $A$ and $B$ are the
two sides of the separator. Note that in order to make this work, we in fact solve a
slightly extended problem, where some arbitrary subset of nodes are already colored, 
and the task is to extend this coloring to a proper $q$-coloring.

As cliques have size at most $q=O(1)$, a separator consisting of $s_\G(n)$ cliques has
$O(s_\G(n))$ nodes, and its nodes can be colored in at most
$q^{O(s_\G(n))}=2^{O(s_\G(n))}$ many ways. The running time obeys the 
recurrence
\[
T(n)=2^{O(s_\G(n))}(T(\beta' n)+\poly(n)),
\]
where $\beta'<1$ is a constant such that $\max(|A|,|B|)+|S|<\beta' n$ is guaranteed by
the separator for $n$ large enough. The recursion solves to $T(n)=2^{O(s_\G(n))}$.
\end{proof}

\section{Concluding Remarks}

We showed how clique-based separators with sub-linear weight can be constructed for various classes of intersection graphs which involve non-fat objects. The main advantage of our approach is that we can solve different problems in the graph classes we study in a uniform manner. There are several natural questions that are left open. Some are listed below.
\begin{itemize}
    \item \textbf{Improving the bound for geodesic disks and adding holes.} Our bound on geodesic disks is directly derived by our result on pseudo-disks. However, geodesic disks are much less general than pseudo-disks (and ``closer'' to regular disks). Hence, one would expect that the optimal weight is closer to $O(\sqrt n)$. If we allow our polygon to have holes, then our approach for geodesic disks no longer works. Indeed, it is easy to see that even after applying our perturbation scheme, the resulting objects can intersect each other more than two times. 
    \item \textbf{Improving the bound for pseudo-disks.} Regarding pseudo-disks, an interesting result \cite{small-pseudodisk} states that in every finite family of pseudo-disks in the plane one can find a ``small'' one, in the sense that it is intersected by only a constant number of disjoint pseudo-disks. This property is also shared by, for instance, convex fat objects. Does this mean that the two graph classes are related in some natural way? If yes, could this connection be exploited to construct separators with better bounds?
\end{itemize}



\bibliography{separator-refs}

\begin{thebibliography}{10}

\bibitem{union-complexity}
B.~Aronov, M.T. {Berg, de}, E.~Ezra, and M.~Sharir.
\newblock Improved bounds for the union of locally fat objects in the plane.
\newblock {\em SIAM Journal on Computing}, 43(2):543--572, 2014.
\newblock \href {https://doi.org/10.1137/120891241}
  {\path{doi:10.1137/120891241}}.

\bibitem{BHKM-vis-graph}
Boaz Ben{-}Moshe, Olaf~A. Hall{-}Holt, Matthew~J. Katz, and Joseph S.~B.
  Mitchell.
\newblock Computing the visibility graph of points within a polygon.
\newblock In {\em Proc. 20th {ACM} Symposium on Computational Geometry}, pages
  27--35. {ACM}, 2004.
\newblock \href {https://doi.org/10.1145/997817.997825}
  {\path{doi:10.1145/997817.997825}}.

\bibitem{rankbased}
Hans~L. Bodlaender, Marek Cygan, Stefan Kratsch, and Jesper Nederlof.
\newblock Deterministic single exponential time algorithms for connectivity
  problems parameterized by treewidth.
\newblock {\em Inf. Comput.}, 243:86--111, 2015.
\newblock \href {https://doi.org/10.1016/j.ic.2014.12.008}
  {\path{doi:10.1016/j.ic.2014.12.008}}.

\bibitem{BonnetR19}
{\'{E}}douard Bonnet and Pawel Rzazewski.
\newblock Optimality program in segment and string graphs.
\newblock {\em Algorithmica}, 81(7):3047--3073, 2019.
\newblock \href {https://doi.org/10.1007/s00453-019-00568-7}
  {\path{doi:10.1007/s00453-019-00568-7}}.

\bibitem{Breen98}
Michael Breen.
\newblock A {Helly}-type theorem for intersections of compact connected sets in
  the plane.
\newblock {\em Geometriae Dedicata}, 71:111--117, 1998.
\newblock \href {https://doi.org/10.1023/A:1005003212822}
  {\path{doi:10.1023/A:1005003212822}}.

\bibitem{C-PTAS-fat}
Timothy~M. Chan.
\newblock Polynomial-time approximation schemes for packing and piercing fat
  objects.
\newblock {\em J. Algorithms}, 46(2):178--189, 2003.
\newblock \href {https://doi.org/10.1016/S0196-6774(02)00294-8}
  {\path{doi:10.1016/S0196-6774(02)00294-8}}.

\bibitem{C-map-graph-mis}
Zhi{-}Zhong Chen.
\newblock Approximation algorithms for independent sets in map graphs.
\newblock {\em J. Algorithms}, 41(1):20--40, 2001.
\newblock \href {https://doi.org/10.1006/jagm.2001.1178}
  {\path{doi:10.1006/jagm.2001.1178}}.

\bibitem{CGP-map-graphs}
Zhi{-}Zhong Chen, Michelangelo Grigni, and Christos~H. Papadimitriou.
\newblock Map graphs.
\newblock {\em J. {ACM}}, 49(2):127--138, 2002.
\newblock \href {https://doi.org/10.1145/506147.506148}
  {\path{doi:10.1145/506147.506148}}.

\bibitem{CW-random-sampling-II}
Kenneth~L. Clarkson and Peter~W. Shor.
\newblock Application of random sampling in computational geometry, {II}.
\newblock {\em Discret. Comput. Geom.}, 4:387--421, 1989.
\newblock \href {https://doi.org/10.1007/BF02187740}
  {\path{doi:10.1007/BF02187740}}.

\bibitem{BergBKK18}
Mark de~Berg, Hans~L. Bodlaender, S{\'{a}}ndor Kisfaludi{-}Bak, and Sudeshna
  Kolay.
\newblock An {ETH}-tight exact algorithm for euclidean {TSP}.
\newblock In {\em 59th {IEEE} Annual Symposium on Foundations of Computer
  Science (FOCS)}, pages 450--461, 2018.
\newblock \href {https://doi.org/10.1109/FOCS.2018.00050}
  {\path{doi:10.1109/FOCS.2018.00050}}.

\bibitem{bbkmz-ethf-20}
Mark de~Berg, Hans~L. Bodlaender, S{\'{a}}ndor Kisfaludi{-}Bak, D{\'{a}}niel
  Marx, and Tom~C. van~der Zanden.
\newblock A framework for {Exponential-Time-Hypothesis}-tight algorithms and
  lower bounds in geometric intersection graphs.
\newblock {\em SIAM J. Comput.}, 49:1291--1331, 2020.
\newblock \href {https://doi.org/10.1137/20M1320870}
  {\path{doi:10.1137/20M1320870}}.

\bibitem{bcko-cgaa-08}
Mark de~Berg, Otfried Cheong, Marc~J. van Kreveld, and Mark~H. Overmars.
\newblock {\em Computational Geometry: Algorithms and Applications (3rd
  Edition)}.
\newblock Springer, 2008.
\newblock URL: \url{https://www.worldcat.org/oclc/227584184}.

\bibitem{DV-cycle-separator}
Hristo Djidjev and Shankar~M. Venkatesan.
\newblock Reduced constants for simple cycle graph separation.
\newblock {\em Acta Informatica}, 34(3):231--243, 1997.
\newblock \href {https://doi.org/10.1007/s002360050082}
  {\path{doi:10.1007/s002360050082}}.

\bibitem{FominLP0Z19}
Fedor~V. Fomin, Daniel Lokshtanov, Fahad Panolan, Saket Saurabh, and Meirav
  Zehavi.
\newblock Decomposition of map graphs with applications.
\newblock In {\em Proc.~46th International Colloquium on Automata, Languages,
  and Programming, (ICALP)}, pages 60:1--60:15, 2019.
\newblock \href {https://doi.org/10.4230/LIPIcs.ICALP.2019.60}
  {\path{doi:10.4230/LIPIcs.ICALP.2019.60}}.

\bibitem{FoxPT10}
Jacob Fox, J{\'{a}}nos Pach, and Csaba~D. T{\'{o}}th.
\newblock A bipartite strengthening of the crossing lemma.
\newblock {\em J. Comb. Theory, Ser. {B}}, 100(1):23--35, 2010.
\newblock \href {https://doi.org/10.1016/j.jctb.2009.03.005}
  {\path{doi:10.1016/j.jctb.2009.03.005}}.

\bibitem{HQ-low-dens-sep}
Sariel Har{-}Peled and Kent Quanrud.
\newblock Approximation algorithms for polynomial-expansion and low-density
  graphs.
\newblock {\em {SIAM} J. Comput.}, 46(6):1712--1744, 2017.
\newblock \href {https://doi.org/10.1137/16M1079336}
  {\path{doi:10.1137/16M1079336}}.

\bibitem{KLPS-pseudo-disk-union}
Klara Kedem, Ron Livne, J{\'{a}}nos Pach, and Micha Sharir.
\newblock On the union of {Jordan} regions and collision-free translational
  motion amidst polygonal obstacles.
\newblock {\em Discret. Comput. Geom.}, 1:59--70, 1986.
\newblock \href {https://doi.org/10.1007/BF02187683}
  {\path{doi:10.1007/BF02187683}}.

\bibitem{KB-hyperbolic-int-graphs}
S{\'{a}}ndor Kisfaludi{-}Bak.
\newblock Hyperbolic intersection graphs and (quasi)-polynomial time.
\newblock In {\em Proc. 31st {ACM-SIAM} Symposium on Discrete Algorithms
  {(SODA)}}, pages 1621--1638, 2020.
\newblock \href {https://doi.org/10.1137/1.9781611975994.100}
  {\path{doi:10.1137/1.9781611975994.100}}.

\bibitem{KisfaludiMZ19}
S{\'{a}}ndor Kisfaludi{-}Bak, D{\'{a}}niel Marx, and Tom~C. van~der Zanden.
\newblock How does object fatness impact the complexity of packing in $d$
  dimensions?
\newblock In {\em 30th International Symposium on Algorithms and Computation,
  {ISAAC} 2019}, volume 149 of {\em LIPIcs}, pages 36:1--36:18, 2019.
\newblock \href {https://doi.org/10.4230/LIPIcs.ISAAC.2019.36}
  {\path{doi:10.4230/LIPIcs.ISAAC.2019.36}}.

\bibitem{Lee-string-sep}
James~R. Lee.
\newblock Separators in region intersection graphs.
\newblock In {\em 8th Innovations in Theoretical Computer Science Conference,
  {ITCS} 2017}, volume~67 of {\em LIPIcs}, pages 1:1--1:8, 2017.
\newblock \href {https://doi.org/10.4230/LIPIcs.ITCS.2017.1}
  {\path{doi:10.4230/LIPIcs.ITCS.2017.1}}.

\bibitem{LT-planar-separator-thm}
Richard~J. Lipton and Robert~Endre Tarjan.
\newblock A separator theorem for planar graphs.
\newblock {\em {SIAM} J. Appl. Math}, 36(2):177--189, 1977.
\newblock \href {https://doi.org/doi/10.1137/0136016}
  {\path{doi:doi/10.1137/0136016}}.

\bibitem{MarxP15}
D{\'{a}}niel Marx and Michal Pilipczuk.
\newblock Optimal parameterized algorithms for planar facility location
  problems using voronoi diagrams.
\newblock In {\em Proc.~23rd Annual European Symposium on Algorithms (ESA)},
  volume 9294 of {\em Lecture Notes in Computer Science}, pages 865--877.
  Springer, 2015.
\newblock \href {https://doi.org/10.1007/978-3-662-48350-3\_72}
  {\path{doi:10.1007/978-3-662-48350-3\_72}}.

\bibitem{m-lecture-notes}
Jir{\'{\i}} Matousek.
\newblock {\em Lectures on Discrete Geometry}, volume 212 of {\em Graduate
  Texts in Mathematics}.
\newblock Springer, 2002.

\bibitem{Matousek14}
Jir{\'{\i}} Matou\v{s}ek.
\newblock Near-optimal separators in string graphs.
\newblock {\em Comb. Probab. Comput.}, 23(1):135--139, 2014.
\newblock \href {https://doi.org/10.1017/S0963548313000400}
  {\path{doi:10.1017/S0963548313000400}}.

\bibitem{MTTV-sep-sphere-packing}
Gary~L. Miller, Shang{-}Hua Teng, William~P. Thurston, and Stephen~A. Vavasis.
\newblock Separators for sphere-packings and nearest neighbor graphs.
\newblock {\em J. {ACM}}, 44(1):1--29, 1997.
\newblock \href {https://doi.org/10.1145/256292.256294}
  {\path{doi:10.1145/256292.256294}}.

\bibitem{PS-incidences}
J\'anos Pach and Micha Sharir.
\newblock Geometric incidences.
\newblock In J\'anos Pach, editor, {\em Towards a Theory of Geometric Graphs
  (Contemporary Mathematics, Vol. 342)}, pages 185--223. Amer. Math. Soc.,
  2004.

\bibitem{small-pseudodisk}
Rom Pinchasi.
\newblock A finite family of pseudodiscs must include a “small” pseudodisc.
\newblock {\em SIAM Journal on Discrete Mathematics}, 28:1930--1934, 10 2014.
\newblock \href {https://doi.org/10.1137/130949750}
  {\path{doi:10.1137/130949750}}.

\bibitem{Poll89}
Ricky Pollack, Micha Sharir, and G\"unter Rote.
\newblock Computing the geodesic center of a simple polygon.
\newblock {\em Discret. Comput. Geom.}, 4(6):611–626, 1989.
\newblock \href {https://doi.org/10.1007/BF02187751}
  {\path{doi:10.1007/BF02187751}}.

\bibitem{RamanRay20}
Rajiv Raman and Saurabh Ray.
\newblock Constructing planar support for non-piercing regions.
\newblock {\em Discrete \& Computational Geometry}, pages 1--25, 2020.

\bibitem{SW-geom-sep}
Warren~D. Smith and Nicholas~C. Wormald.
\newblock Geometric separator theorems {\&} applications.
\newblock In {\em 39th Annual Symposium on Foundations of Computer Science
  (FOCS)}, pages 232--243. {IEEE} Computer Society, 1998.
\newblock \href {https://doi.org/10.1109/SFCS.1998.743449}
  {\path{doi:10.1109/SFCS.1998.743449}}.

\end{thebibliography}

\newpage




\end{document}